\documentclass[a4paper, twocolumn, superscriptaddress, showkeys, floatfix,11pt,tightenlines]{quantumarticle}
\pdfoutput=1
\usepackage{silence}
\usepackage[utf8]{inputenc}
\usepackage[english]{babel}
\usepackage[T1]{fontenc}
\usepackage{array}
\usepackage{multirow}
\usepackage{dcolumn}
\usepackage{hyperref}
\usepackage[pdftex]{graphicx}
\usepackage[normalem]{ulem}
\usepackage{url}
\usepackage{braket}
\usepackage{here}
\usepackage{amsmath}
\usepackage{amsfonts}
\usepackage{euscript}
\usepackage{mathrsfs}
\usepackage{amssymb}
\usepackage{amsthm}
\usepackage{mathcomp}
\usepackage{scalefnt}
\usepackage{dcolumn}
\usepackage{float}
\usepackage[ruled,linesnumbered,noend]{algorithm2e}
\usepackage[FIGTOPCAP,hang,nooneline]{subfigure}
\usepackage{cite}
\usepackage{tikz}
\usetikzlibrary{decorations.pathreplacing}
\usetikzlibrary{trees,quotes,angles}
\usepackage{xcolor}
\usepackage{caption}
\usepackage{subcaption}
\usepackage{tabularray}
\usepackage{amsthm}
\newtheorem{proposition}{Proposition}

\newcommand{\ctext}[1]{\raise0.2ex\hbox{\textcircled{\scriptsize{#1}}}}

\begin{document}
\title{Efficient Preparation of Graph States using the Quotient-Augmented Strong Split Tree}

\author{Nicholas Connolly$^*$}
\affiliation{Okinawa Institute of Science and Technology Graduate University, Onna-son, Kunigami-gun, Okinawa, 904-0495, Japan}
\email{nicholas.connolly@oist.jp}

\author{Shin Nishio$^*$}
\affiliation{Graduate School of Science and Technology, Keio University, Yokohama, Kanagawa, 223-8522, Japan}
\affiliation{Department of Physics \& Astronomy, University College London, London, WC1E 6BT, United Kingdom}
\email{shin.nishio@keio.jp}

\author{Dan E. Browne}
\affiliation{Department of Physics \& Astronomy, University College London, London, WC1E 6BT, United Kingdom}
\email{d.browne@ucl.ac.uk}

\author{William John Munro} 
\email{bill.munro@oist.jp}
\affiliation{Okinawa Institute of Science and Technology Graduate University, Onna-son, Kunigami-gun, Okinawa, 904-0495, Japan}
\affiliation{National Institute of Informatics, 2-1-2 Hitotsubashi, Chiyoda-ku, Tokyo, 101-8430, Japan} 
\author{Kae Nemoto} 
\email{kae.nemoto@oist.jp}
\affiliation{Okinawa Institute of Science and Technology Graduate University, Onna-son, Kunigami-gun, Okinawa, 904-0495, Japan}
\affiliation{National Institute of Informatics, 2-1-2 Hitotsubashi, Chiyoda-ku, Tokyo, 101-8430, Japan} 
\begin{flushleft}
\def\thefootnote{*}\footnotetext{These authors contributed equally to this work.}\def\thefootnote{\arabic{footnote}}
\end{flushleft}
\maketitle
\begin{abstract}
Graph states are a key resource for measurement-based quantum computation and quantum networking, but state-preparation costs limit their practical use. 
Graph states related by local complement (LC) operations are equivalent up to single-qubit Clifford gates; one may reduce entangling resources by preparing a favorable LC-equivalent representative. However, exhaustive optimization over the LC orbit is not scalable.
We address this problem using the split decomposition and its quotient-augmented strong split tree (QASST). For several families of distance-hereditary (DH) graphs, we use the QASST to characterize LC orbits and identify representatives with reduced controlled-Z count or preparation circuit depth. We also introduce a split-fuse construction for arbitrary DH graph states, achieving linear scaling with respect to entangling gates, time steps, and auxiliary qubits. Beyond the DH setting, we discuss a generalized divide-and-conquer split-fuse strategy and a simple greedy heuristic for generic graphs based on triangle enumeration. Together, these methods outperform direct implementations on sufficiently large graphs, providing a scalable alternative to brute-force optimization. 
\end{abstract}


\section{Introduction}
\label{sect:introduction}
Graph states are a fundamental resource in quantum information processing, with prominent applications in measurement-based quantum computation (MBQC), quantum error correction, and photonic quantum networks~\cite{raussendorf2001one,hein2004multiparty}. They also provide a convenient graphical representation of the entanglement structure underlying MBQC protocols~\cite{hein2004multiparty}.
As experimental platforms continue to scale up in size, the efficient preparation of large graph states has become a central practical challenge.
Preparation costs are typically quantified by the number of two-qubit entangling gates, usually controlled-Z (CZ) gates, and by the circuit depth required to implement them.

A key structural feature of graph states is their equivalence under \emph{local complement} (LC) operations ~\cite{van2004graphical,van2004efficient,hein2006entanglement}.
Since LC transformations can be implemented using only local Clifford gates, which are typically far less costly than entangling gates, it is often advantageous to prepare an LC-equivalent representative of a target graph state that minimizes preparation resources, and subsequently transform it into the desired state using local operations.
This observation has motivated extensive work on identifying optimal representatives within LC equivalence classes of graph states~\cite{adcock2020mapping,cabello2011optimal,sharma2026minimising}.
Typically, this is done by examining the \textit{LC orbit} of the corresponding graph; that is, the equivalence class of graphs related by LC operations, a type of graph transformation.
Hence, LC orbits provide a convenient way of enumerating equivalent graph states.

The principal difficulty in this approach is that LC orbits grow rapidly with the number of qubits.
Exhaustive enumeration is feasible only for small graph states, and even determining the size of an LC orbit for an arbitrary graph state is a computationally hard problem~\cite{dahlberg2020counting}.
While numerical methods have been used to classify LC orbits and identify optimal representatives for graph states with up to 12 qubits~\cite{adcock2020mapping,cabello2011optimal}, such approaches do not scale to larger systems relevant for practical cases.
This suggests that scalable optimization should avoid brute-force orbit searches and instead exploit structural features of the graph itself. Recent work has similarly shown that specific graph-state families can exhibit strong local-equivalence structure; for example, circle graph states were recently shown to be closed under $r$-local complement~\cite{hahn2026structurecirclegraphstates}.

In this work, we take this structural approach precisely. Our starting point is the split decomposition of a graph ~\cite{cunningham1982decomposition,cunningham1980combinatorial} and its representation as a \emph{quotient-augmented strong split tree} (QASST)~\cite{connolly2026local}. The QASST organizes a graph into smaller quotient graphs connected by an underlying tree structure called a strong split tree~\cite{gioan2007dynamic,gioan2012split}, while retaining enough information to reconstruct the original graph. This viewpoint is useful for two reasons. First, it provides a compact way to describe large portions of an LC orbit in terms of locally equivalent quotient graphs rather than whole graphs~\cite{connolly2026local}. Second, it naturally suggests a preparation strategy in which one prepares smaller intermediate graph states associated with the quotient graphs and then combines them using fusion operations~\cite{browne2005resource}. In this sense, the QASST serves as a common framework for both orbit analysis and graph-state preparation.

Using this structure, we develop two complementary methods for optimizing graph-state preparation. We focus first on \emph{distance-hereditary} (DH) graphs, a class that is especially well-suited to this framework. DH graphs are closed under local complementation~\cite{bouchet1988transforming}, and their split decomposition contains only star and complete quotient graphs~\cite{cunningham1980combinatorial}, making the QASST particularly tractable. Moreover, the DH class includes several graph families relevant to quantum communication and networking, such as complete multipartite graphs, clique-stars, and repeater graphs~\cite{azuma2015all}. First, for several structured DH families, we analytically characterize their LC orbits using the QASST description and identify representatives that minimize either the number of CZ gates or the maximum vertex degree, which bounds the minimum preparation depth~\cite{cabello2011optimal,vizing1964estimate}. This yields provably optimal constructions for complete bipartite graphs, complete multipartite graphs, clique-stars, and related repeater-type graphs without requiring numerical searches over the full LC orbit~\cite{connolly2026local}.
Second, for arbitrary DH graph states, we introduce a split-fuse preparation method based on the split decomposition, which can be efficiently computed~\cite{charbit2012linear,dahlhaus1994efficient,dahlhaus2000parallel} in linear time, and Type-II fusion operations~\cite{browne2005resource}. In this construction, the target graph state is assembled from quotient graph states specified by the QASST. Because all quotient graphs of a DH graph may be initialized as stars, the resulting protocol achieves linear scaling in the number of CZ gates, time steps, and auxiliary qubits with the system size.

We also investigate how this structural perspective extends beyond the DH setting. For generic graphs, the split decomposition still exists, but prime quotient graphs appear and prevent the same fully analytic treatment. This leads to a generalized split-fuse strategy in which star and complete quotient graphs are prepared in optimized form, while prime quotient graphs are handled either directly or with a simple greedy heuristic based on local complement. We propose a simple greedy heuristic that reduces edge count without requiring special structural assumptions on the input graph based on triangle enumeration, which is more efficient than searching the orbit.  

\section{Review of Graphs and Graph States}
\label{sect:review_of_graphs_and_graph_states}

In this section, we open with a discussion of distance-hereditary graphs, an important class of graphs for which we will propose a tailored graph state preparation protocol. For a general overview of the basics and notation of graphs, see Appendix~\ref{app:graph_basic}.
Then we formally define graph \emph{local complement}, an operation on graphs that can be used to describe certain local Clifford operations on graph states.
Following this, we introduce the \emph{quotient-augmented strong split tree} (QASST), our primary tool for characterizing graphs.
Finally, we introduce graph states, describing their preparation and various graph state operations, concluding with a discussion of graph state optimization. 

\begin{figure*}[t]
\begin{center}
\begin{tabular}{ccccccc}
\includegraphics[width=0.12\textwidth,page=1]{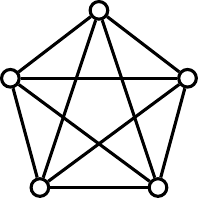}&\includegraphics[width=0.12\textwidth,page=2]{Figures/Section_Review.pdf}&\includegraphics[width=0.11\textwidth,page=3]{Figures/Section_Review.pdf}&\includegraphics[width=0.11\textwidth,page=4]{Figures/Section_Review.pdf}&\includegraphics[width=0.11\textwidth,page=5]{Figures/Section_Review.pdf}&\includegraphics[width=0.11\textwidth,page=6]{Figures/Section_Review.pdf}&\includegraphics[width=0.11\textwidth,page=7]{Figures/Section_Review.pdf}\\
$K_5$&$S_5$&$K_{3,2}$&$K_{2,2,2,2}$&$CS^1_{2,2,2,2}$&$R_4$&$MR_{3,3,3,3}$
\end{tabular}
\end{center}
\caption{
Examples of the special families of graphs we characterize, all of which are distance-hereditary.
}
\label{fig:graph_examples}
\end{figure*}

\begin{figure*}[t]
\centering
\begin{tabular}{ccc}
\includegraphics[width=0.55\linewidth,page=8]{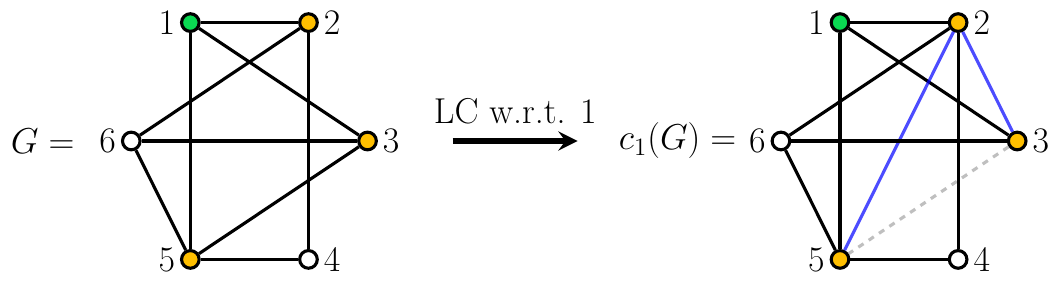}&&\includegraphics[width=0.35\linewidth,page=9]{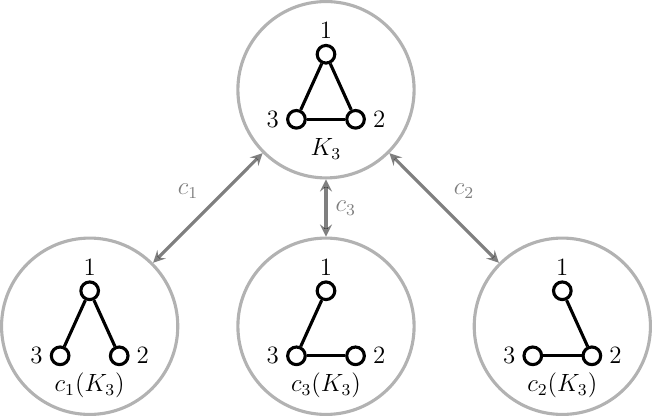}\\
\\
(a) Graph Local Complement&&(b) LC Orbit of $\mathcal{O}(K_3)$
\end{tabular}
\caption{
(a) The effect of local complement on a graph $G$ with respect to vertex $1\in V(G)$ (green vertex). The edges between the neighbors of 1 (yellow vertices) are complemented: existing edges are deleted, and missing edges are added.
(b) The LC orbit of the complete graph on three vertices, locally equivalent to a star with 2 spokes.
}
\label{fig:local_complement_example}
\end{figure*}

\subsection{Distance Hereditary Graphs}

A graph is called \textit{distance-hereditary}~\cite{howorka1977characterization} (DH) if all connected induced subgraphs preserve distance between vertices.
DH graphs have been characterized in a few different ways, including those graphs which can be constructed recursively via twins~\cite{bandelt1986distance}; completely separable graphs~\cite{hammer1990completely}; totally decomposable graphs~\cite{cunningham1982decomposition}; and graphs of rank width 1\cite{oum2005rank}.
Each of these characterizations is equivalent and will be useful for examining the LC orbit of DH graphs in Section~\ref{sect:classifying_LC_oribts}.

DH graphs play an important role in network communication theory, where the distance-hereditary property is closely related to the robustness of a network against vertex loss.
There are many types of distance-hereditary graph, but we are particularly interested in the special families illustrated in Figure~\ref{fig:graph_examples}.
These include: \textit{complete}, \textit{star}, \textit{complete bipartite}, \textit{complete-multipartite}, \textit{clique-star}, \textit{repeater}, and \textit{multi-leaf repeater}.
Many of these graphs, such as repeaters~\cite{azuma2015all}, are already used in quantum protocols.
Furthermore, these graphs have particularly simple descriptions thanks to their rich symmetries.
The formal definitions for these special graphs are included in Appendix~\ref{app:special_graph_classes}.

\subsection{Graph Local Complement}
\label{sect:graph_local_complement}

Certain operations on graph states are modeled by corresponding operations on graphs.
In particular, we are concerned with \emph{local Clifford} operations, a type of single-qubit operation which modifies the connectivity of a graph state while preserving the entanglement.
\emph{Local Clifford} operations refer to those single-qubit operations obtainable by applying Hadamard, phase, or Pauli gates on individual qubits.
The graph operation known as \emph{local complement} on graphs can be implemented on graph states with Pauli gates, which are a subset of local Clifford operations~\cite{van2004graphical,van2004efficient,hein2006entanglement}.

\begin{figure*}[t]
\centering
\small
\begin{tabular}{ccccc}
\includegraphics[width=0.3\linewidth,page=13]{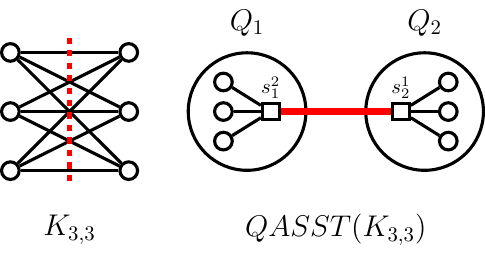}&&\includegraphics[width=0.3\linewidth,page=14]{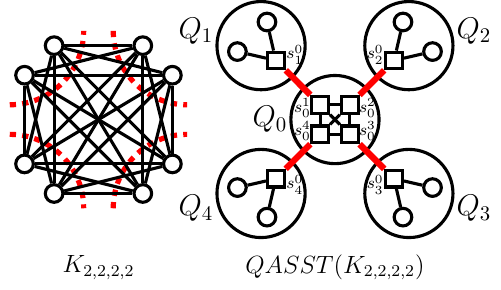}&&\includegraphics[width=0.3\linewidth,page=15]{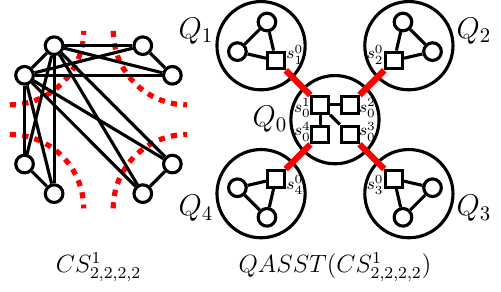}\\
(a) Complete Bipartite&&(b) Complete Multipartite&&(c) Clique-Star 
\end{tabular}
\normalsize
\caption{
The split decompositions for the three special families of DH graph that we consider. 
$K_{n,m}$ splits into two quotient graphs, while $K_{n_1,\cdots,n_k}$ and $CS^r_{n_1,\cdots,n_k}$ always split into $k+1$ quotient graphs. We adopt the convention of labeling the central quotient graph $Q_0$ and the others $Q_1,\cdots,Q_k$ matching the index of the vertex group.
The LC orbits of these graphs have been fully classified based on the symmetries of quotient graphs in the QASST.}
\label{fig:special_graph_QASSTs}
\end{figure*}

Originally introduced by Bouchet~\cite{bouchet1993recognizing,bouchet1991efficient,bouchet1988graphic}, \emph{local complement} is a graph operation defined with respect to a single vertex that induces the following transformation.
Given a graph $G$ and a vertex $v\in V(G)$, the \emph{local complement of $G$ with respect to $v$} is the graph $c_v(G)$ obtained by replacing the neighborhood of $v$ with its edge complement (Figure~\ref{fig:local_complement_example}).
Note that this operation is self-inverse: a second application of local complement with respect to the same vertex recovers the original graph.
Bouchet studied local complements extensively in his research and used them to characterize a number of interesting graph families, including trees (graphs containing no cycles)~\cite{bouchet1988transforming} and circle graphs (graphs defined via the chords in a circle)~\cite{bouchet1987reducing}.

Let ${\mathcal G}_n$ denote the set of connected, simple, labeled graphs with $n$ vertices. Local complement can be regarded as a map from ${\mathcal G}_n$ to itself, and so we introduce the function notation $c_v:{\mathcal G}_n\rightarrow{\mathcal G}_n$ to denote the operation of local complement on a graph $G\in{\mathcal G}_n$ with respect to a vertex $v\in V(G)$.
Two graphs which are related by a sequence of zero or more local complements are said to be \textit{locally equivalent} or \textit{LC-equivalent}, and we will write $G\cong_{\text{LC}}G'$ for two graphs related in this way.
In this way, local complement defines an \textit{equivalence relation} on ${\mathcal G}_n$ (i.e.~a relation which is reflexive, symmetric, and transitive).
The equivalence classes with respect to this relation are called \textit{LC orbits}, and these form a partition of ${\mathcal G}_n$.
For a given graph $G$, let ${\mathcal O}(G)$ represent its LC orbit.
Note that $G\cong_{\text{LC}}G'$ if and only if ${\mathcal O}(G)=\mathcal{O}(G')$.
Figure~\ref{fig:local_complement_example} shows a simple example of the LC orbit of the complete graph on three vertices.

\begin{figure*}[htbp]
\centering
\scalebox{0.88}{
\begin{tabular}{ccccccc}
\includegraphics[scale=1,page=20]{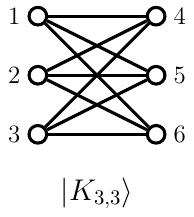}&&\includegraphics[scale=0.65,page=21]{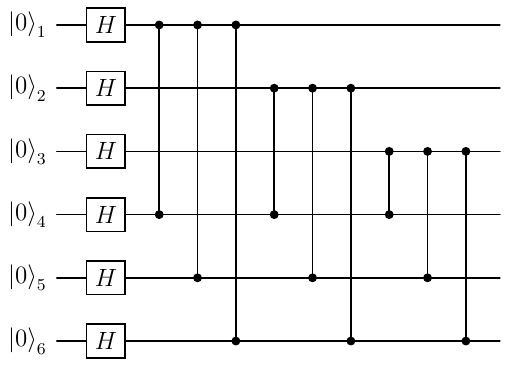}&&\includegraphics[scale=1,page=22]{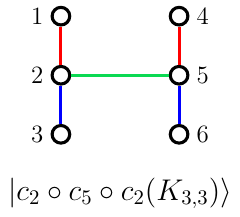}&&\includegraphics[scale=0.65,page=23]{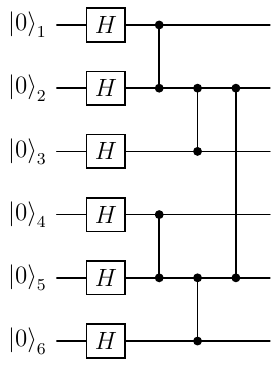}\\
\\
(a) Graph State&&(b) Sequential Circuit&&(c) LC-equivalent State&&(d) Minimal Circuit
\end{tabular}
}
\caption{
An example of two LC-equivalent graph states and their preparation circuits. Vertices denote qubits prepared in the $|+\rangle=H|0\rangle$ state, and edges denote entanglement via a controlled-Z (CZ) gate.
(a) The graph state corresponding to the complete bipartite graph $K_{3,3}$.
(b) The preparation circuit for $|K_{3,3}\rangle$ with CZ operations ordered sequentially.
(c) The graph state for a binary-star, locally equivalent to $|K_{3,3}\rangle$ via an edge-pivot. This graph has chromatic index $\chi'(G)=3$, and edges are colored according to a minimal edge-coloring.
(d) A minimal-depth preparation circuit corresponding to this edge-coloring, wherein edges of the same color correspond to CZ operations that are performed simultaneously.
This is an optimal graph state locally equivalent to $|K_{3,3}\rangle$ (non-unique).
}
\label{fig:graph_coloring_schedule}
\end{figure*}

\subsection{Split Decompositions and the QASST}
\label{sect:split_decomposition}

A \textit{split} of a graph $G$ is a bipartition of the vertex set $V(G)=U_1\sqcup U_2$ such that the edges passing between $U_1$ and $U_2$ induce a complete bipartite subgraph.
A split is called trivial if one side consists of a single vertex.
A graph with no nontrivial splits is called prime; the 5-cycle is the smallest example of a prime graph.
Two splits are said to cross if each side of one split has a nonempty intersection with the sides of the other split.
A nontrivial split is called \emph{strong} if it crosses no other split; trivial splits are defined to be strong.
Figure~\ref{fig:special_graph_QASSTs} shows several examples of graphs with strong splits indicated by dashed red lines; the edges crossing these dashed lines induce complete bipartite subgraphs.

Any graph $G$ can be decomposed into a collection of \textit{quotient graphs} by collapsing all of its strong splits, a process known as the \textit{split decomposition} and originally introduced by Cunningham~\cite{cunningham1982decomposition,cunningham1980combinatorial}.
By collapsing splits into pairs of \textit{split-nodes} representing all-to-all connections between neighbors (indicated by squares in Figure~\ref{fig:special_graph_QASSTs}), this process is naturally described by a tree whose vertices are labeled by the quotient graphs of $G$. In the figure, the quotient graphs are labeled $Q_1,Q_2,\cdots$ as shown, with red edges between quotient graphs representing collapsed splits between pairs of split-nodes. The idea of representing the split decomposition by a graph-labeled tree was formalized by Gioan and Paul~\cite{gioan2007dynamic,gioan2012split} and expanded on in our previous work~\cite{connolly2026local}.
The non-split-nodes are called \textit{leaf-nodes} and are denoted by circles; these are the vertices inherited from the original graph $G$.

The split decomposition of a graph $G$ gives two pieces of information: a strong split tree whose edges are in bijection with the strong splits of $G$ (excluding trivial splits) and a list of quotient graphs $Q_1,\cdots,Q_k$ in bijection with the vertices of this tree.
Associated to each edge in the strong split tree, we also designate a pair of split-nodes corresponding to nodes $s_i^j$ and $s_j^i$ from two adjacent quotient graphs $Q_i$ and $Q_j$ in the tree. The subscript on $s_i^j$ represents the index of the quotient graph $Q_i$ which contains it, while the superscript represents the index of the quotient graph $Q_j$ belonging to its partner.
Figure~\ref{fig:special_graph_QASSTs} shows this explicitly for three examples.
We refer to this entire object as the \textit{quotient-augmented strong-split tree} of $G$, or $\textit{QASST}(G)$.

The QASST of any graph can be obtained from the split decomposition and quotient graphs, as shown for three examples of DH graphs in Figure~\ref{fig:special_graph_QASSTs}.
Because all strong splits are collapsed in the construction, all quotient graphs are guaranteed to be classified as prime, star, or complete.
The QASST contains equivalent information to the original graph, which can be reconstructed by replacing pairs of adjacent split-nodes with all-to-all connections between neighbors.
The advantage of viewing a graph in terms of its QASST is that questions about the original graph can be studied in terms of the smaller quotient graphs, potentially dividing computational algorithms into smaller problems.
We leverage this fact in Section~\ref{sect:the_split_fuse_method} to introduce a graph state preparation technique based on the QASST.

There is an important relationship between local equivalence and the split decomposition.
In particular, Bouchet proved that the strong splits in a graph are invariant under local complements (Lemma 2.1 of~\cite{bouchet1987reducing}).
This means that the split decompositions of two locally equivalent graphs are described by the same QASST, but with a different combination of quotient graphs.
We showed in~\cite{connolly2026local} that the quotient graphs of two locally equivalent graphs must themselves be locally equivalent.
This fact extends to the entire LC orbit, and hence the underlying tree structure of the QASST is an invariant of the LC equivalence class.
Based on this fact, the LC orbit may be characterized in terms of locally equivalent quotient graphs used with the QASST. 

\subsection{Graph States}
\label{sect:graph_states}

Graph states are a type of quantum state fundamental to measurement-based quantum computing~\cite{raussendorf2001one}. 
The preparation circuit for a graph state is described exactly by a simple graph~\cite{hein2004multiparty}, wherein vertices denote qubits in the $|+\rangle=H|0\rangle$ state and edges represent entanglement through a CZ gate.
We will use the notation $|G\rangle$ to denote the graph state corresponding to a graph $G$.
\begin{eqnarray}
|G\rangle&=&\prod_{(u,w)\in E(G)}CZ_{u,w}\bigotimes_{v\in V(G)}|+\rangle_v
\end{eqnarray}
Figure~\ref{fig:graph_coloring_schedule} shows an example of the preparation circuit for the graph state corresponding to $K_{3,3}$, with CZ gates ordered sequentially.

Two-qubit CZ gates are generally considered the most resource expensive operations in the preparation of a graph state.
However, single-qubit gates such as local Clifford operations are typically inexpensive by comparison.
Two graph states which are related by some sequence of local Clifford operations are said to be local Clifford (LC) equivalent, and such graph states are considered to exhibit the same entanglement.
Hence, preparing an LC-equivalent graph state with fewer CZ gates, and then transforming to a target graph state via local Clifford transformations is one practical technique for reducing resource costs.

The local Clifford operations on a graph state $|G\rangle$ corresponding to a local complement operation on the graph $G$ with respect to a vertex $v\in V(G)$ can be stated explicitly as rotations on the Bloch sphere for the qubits corresponding to $v$ and its neighbors.
Specifically, the unitary operation which transforms $|G\rangle$ into $|c_v(G)\rangle$ is given by the operation~\cite{hein2006entanglement,dahlberg2018transforming}
\begin{eqnarray}\label{eq:LC}
U^G_v&=&\exp{\left(-i\frac{\pi}{4}X_v\right)}\prod_{u\in N(v)}\exp{\left(i\frac{\pi}{4}Z_u\right)}\nonumber,\\
&&
\end{eqnarray}
wherein
\begin{eqnarray}
|c_v(G)\rangle&=&U_v^G|G\rangle
\end{eqnarray}
Similarly, Pauli measurements on the qubits in a graph state involve local Clifford operations followed by a vertex deletion. These can be described in terms of \textit{vertex minors}. A description of this relationship and explicit formulas for such measurements are summarized in Appendix~\ref{app:vertex_minors_and_measurements}.

\textit{Fusions} are another type of operation which can be used to combine two graph states into a single, larger graph state~\cite{browne2005resource}. There are two types of fusion operations, but in this paper we are interested specifically in Type-II fusions, specifically the adaptation introduced in~\cite{lee2023graph}. This Type-II fusion is implemented by a Bell-state measurement preceded by a Hadamard gate on one of the qubits, and followed by the removal of the measured qubits.
Graphically, a Type-II fusion between two qubits $q_1$ and $q_2$ from two different graph states $|G_1\rangle$ and $|G_2\rangle$ results in a new graph state $|G\rangle$, wherein each of the neighbors of $q_1$ is connected to all of the neighbors of $q_2$ (and vice versa) and then $q_1$ and $q_2$ are deleted~\cite{lee2023graph}.
A simple example of this is shown in Figure~\ref{fig:type-II_fusion}.

\begin{figure}[t]
\centering
\includegraphics[width=\linewidth,page=11]{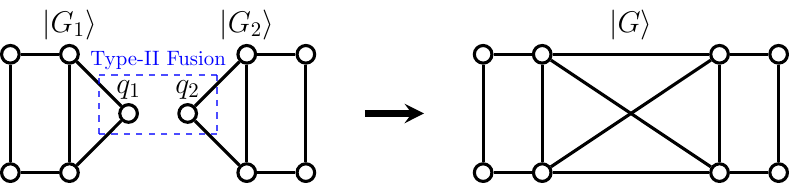}
\caption{
An example of a Type-II fusion between qubits $q_1$ and $q_2$ from two different graph states $|G_1\rangle$ and $|G_2\rangle$.
}
\label{fig:type-II_fusion}
\end{figure}

\subsection{Minimizing Preparation Resources}

If $|G\rangle$ is the graph state corresponding to a graph $G$, the resources needed to implement the preparation circuit for $|G\rangle$ (i.~e. the number of qubits, the number of gates, and the circuit depth) can be stated in terms of the parameters of $G$.
The number of qubits is the number of vertices $|V(G)|$ and the number of CZ gates is the number of edges $|E(G)|$.
If all CZ gates are implemented sequentially, $|E(G)|$ is also the maximum circuit depth.
However, it is far more efficient to rearrange the circuit so that many CZ gates are performed simultaneously, provided that each qubit is not involved in more than one operation at a time.
This reduces to the well-known graph edge-coloring problem.

Scheduling all of the CZ operations in the preparation of $|G\rangle$ is equivalent to finding an edge-coloring of $G$: that is, a choice of color for each edge such that no vertex is incident to more than one edge of a given color. Each color represents one time step in the circuit, wherein all CZ operations corresponding to edges of the same color may be performed simultaneously. The minimum number of required colors for an edge-coloring of $G$, known as the chromatic index $\chi'(G)$, coincides with the minimum circuit depth~\cite{cabello2011optimal}. There is a simple upper bound $\chi'(G)\leq\Delta(G)+1$ due to Vizing's well-known theorem~\cite{vizing1964estimate}, and hence the number of required time steps depends on the maximum vertex degree of the graph.
The second graph state of Figure~\ref{fig:graph_coloring_schedule} illustrates a simple example of the relationship between this scheduling problem and the chromatic index.
Although the minimal edge-coloring problem is NP-complete in general, there exists a polynomial-time algorithm based on Vizing's theorem for finding an edge-coloring of $G$ with $\Delta(G)+1$ colors (and hence a scheduling for the circuit of $|G\rangle$ in $\Delta(G)+1$ time steps).

Since LC-equivalent graph states exhibit the same entanglement structure, the resources required to implement a target graph state may be reduced by initially preparing an optimal LC-equivalent state, and then transforming it into the target state using local Clifford operations, which are physically much less expensive than performing CZ gates.
This technique has been used in the literature to optimize the preparation of graph states by first choosing a representative of the LC orbit which is minimal with respect to $|E(G)|$ or $\Delta(G)$~\cite{cabello2011optimal,sharma2026minimising}.
These correspond to reducing the number of CZ gates or the circuit depth, respectively, noting that these may not always coincide.
However, these previous optimizations are based on numerical comparisons of the graphs across the entire LC orbit and have only been used to characterize graph states with up to 12 or 16 qubits, respectively.
Given the complexity of the LC orbit, scaling this numerical technique to larger graph states is impractical in general, but the core idea is still applicable for restricted families of graphs whose structure can be described generally.
This is one of the primary motivations for the study of LC-equivalent graph states which we build on in this paper.

\section{Classifying LC Orbits}
\label{sect:classifying_LC_oribts}

The size of the LC orbit grows rapidly, making it challenging in general to count the number of graphs locally equivalent to a given graph.
For graphs with up to 9 qubits, computational methods have been used to enumerate LC orbits exhaustively~\cite{adcock2020mapping}, and optimal preparations across the LC orbit have been identified for graph states with up 12 qubits~\cite{cabello2011optimal}.
However, Dahlberg et. al. proved that counting the size of the LC orbit for an arbitrary graph is \#P-complete~\cite{dahlberg2020counting}.
Even so, general formulas have been derived for some special families of graphs, including path graphs and cycle graphs~\cite{bouchet1993recognizing}.
In a recent work, the authors of this paper have built on these findings by fully classifying the LC orbits of several families of graphs, including complete bipartite graphs, complete multi-partite graphs, clique-stars, and multi-leaf repeater graphs~\cite{connolly2026local}. The following discussion is based on these results.

We proceed in this section by summarizing the mathematical tools used to perform this classification.
First, we show how the QASST can be used to classify the LC orbit of a graph, with particular attention to DH graphs.
Next, we leverage these classifications to identify optimal graphs across the LC orbit for several special cases.
Finally, we describe a short heuristic algorithm for navigating the LC orbit of any graph.

\subsection{LC Orbit Classification for DH Graphs}
\label{sect:LC_orbit_classification_for_DH_graphs}

Locally equivalent graphs have the same strong splits and locally equivalent quotient graphs.
Hence, fixing a split decomposition and counting all the possible ways locally equivalent quotient graphs may be substituted into the QASST provides an upper bound on the size of the LC orbit.
However, this upper bound is not tight in general: there exist non-locally equivalent graphs with the same strong splits and locally equivalent quotient graphs.
Complete multipartite graphs and clique-stars are one such counterexample (Figure~\ref{fig:special_graph_QASSTs}).
We introduce the term \textit{QASST equivalent} to refer to graphs related in this way, noting that this is weaker than LC equivalence.
Although counting the number of QASST equivalent graphs requires knowledge about the LC orbits for quotient graphs, this has the practical advantage of avoiding the enumeration of the entire LC orbit for a complicated graph.
For distance-hereditary graphs in particular, we shall see that counting QASST equivalent graphs is far simpler than counting LC-equivalent graphs directly.

Bouchet further proved that the property of being distance-hereditary is preserved under local complements (Corollary 4.2 of~\cite{bouchet1988transforming}).
Hence, we next consider the relationship between DH graphs and the QASST.
Cunningham showed that any connected graph has a unique split decomposition into quotient graphs which are complete, star, or prime (Theorem 3 of~\cite{cunningham1982decomposition}).
A graph which decomposes into only star and complete quotient graphs is called \textit{completely separable}~\cite{hammer1990completely}; this is equivalent to the distance-hereditary condition.
Star and complete graphs are themselves locally equivalent: $K_n\cong_{\text{LC}}S_{n-1}$. Furthermore, the size of the LC orbit is known to be $|{\mathcal O}(K_n)|=n+1$, consisting of $K_n$ and $n$ different star graphs with a different choice of center (see part (b) of Figure~\ref{fig:local_complement_example}).
Hence, the size of the QASST equivalence class of a DH graph with $k$ star or complete quotient graphs $Q_1,\cdots,Q_k$ is bounded above by $\prod_{i=1}^k(|V(Q_i)|+1)$, but certain combinations of adjacent quotient graphs which do not preserve strong splits are disallowed. These invalid cases depend specifically on the tree structure of the QASST.

A full discussion of how we derived explicit formulas for the size of the QASST equivalence class in the special cases of complete bipartite graphs, complete multipartite graphs, and clique-stars is included in~\cite{connolly2026local}, but we exclude the technical details from this paper.
Combining these with combinatorial arguments to enumerate symmetries in the graph structure, we obtained further explicit formulas for the sizes of the LC orbits for these cases.
For the complete bipartite graph, this is 
\begin{eqnarray}
|{\mathcal O}(K_{n,m})|&=&nm+n+m+3,
\end{eqnarray}
where the symmetries are described in Table~\ref{tab:complete_bipartite_LC_orbit}.
The formulas for complete multipartite graphs and clique-stars are given in Appendix~\ref{app:formulas_for_CM_and_CS}.
These cases also account for multi-leaf repeater graphs, which we proved are locally equivalent to a complete multipartite graph when $k$ is even and to a clique-star when $k$ is odd.
Examples of the split decomposition for graphs in these three families are shown in Figure~\ref{fig:special_graph_QASSTs}.

\subsection{Optimizations across the LC Orbit}
\label{sect:optimizations_across_the_LC_orbit}

With a full understanding of the LC orbit of a graph, one may choose a representative which is optimal with respect to some parameter across the entire equivalence class.
This can be accomplished using exhaustive numerical searches through the LC orbit such as those of~\cite{cabello2011optimal,adcock2020mapping}, which identified optimal representatives for graph states with up to 12 qubits.
We generally choose a graph $G$ with a minimal $|E(G)|$ to minimize the number of CZ gates in the preparation circuit of $|G\rangle$, or a minimal $\Delta(G)$ to minimize the circuit depth.
Graphs minimizing these two parameters may coincide, as in the LC orbit of a complete bipartite graph (Table~\ref{tab:complete_bipartite_LC_orbit}), but this behavior is not guaranteed to occur in general.
One must consider the trade-offs between CZ gate count and circuit depth when optimizing with respect either parameter.
Even so, searching the orbit for a graph with fewer edges also has the effect of reducing the total degree of the graph. Hence, minimizing $|E(G)|$ also tends to reduce $\Delta(G)$, even if a global minimum is not achieved.

Numerical searches of the LC orbit have been useful for small cases, but the growing size of the orbit makes this approach impractical for larger graphs.
By contrast, the QASST-characterization of the LC orbit which we introduced in~\cite{connolly2026local} exploits the fact that many representatives of the orbit can be described in terms of the symmetries of their quotient graphs.
Counting distinct graphs reduces to counting symmetries, and graphs with the same symmetries have the same parameters.
The complete bipartite graph provides a simple example illustrating this idea: any graph in the LC orbit of $K_{n,m}$ falls into one of the six symmetry cases described in Table~\ref{tab:complete_bipartite_LC_orbit}.
Hence, identifying an optimal representative of the LC orbit reduces to comparing a handful of symmetry classes rather than exploring the entire orbit.
Furthermore, these symmetry classes are independent of the number of vertices, provided $n,m\geq2$.

A representative of ${\mathcal O}(K_{n,m})$ with a minimal number of edges and minimal $\Delta(G)$ may be obtained by choosing any graph in the sixth symmetry class of Table~\ref{tab:complete_bipartite_LC_orbit}.
Note that this fact was already established in~\cite{tzitrin2018local}; our goal here is merely to illustrate the utility of the LC orbit characterization for DH graphs using the QASST.
We also provide an explicit sequence of local complements between $K_{n,m}$ and any graph in one of these symmetry classes. In the graph state case, these can be used to compute an explicit sequence of local Clifford operations which converts an optimal graph state to a target graph state.
In addition to this, Table~\ref{tab:bipartite_fixed_parameters} provides a summary of the parameters for an optimal representative of the LC orbit of $K_{n,m}$ for various values of $n$ and $m$, along with the size of the full LC orbit.
The second graph state example of Figure~\ref{fig:graph_coloring_schedule} shows an optimal representative from ${\mathcal O}(K_{3,3})$.

\begin{table*}
\centering
\scalebox{0.9}{
\begin{tabular}{|cc|cc|c|c|c|c|}
\hline
$Q_1$ sym.&(count)&$Q_2$ sym.&(count)&total&$|E(G)|$&$\Delta(G)$&Trans. from $K_{n,m}$\\
\hline
star-center&1&star-center&1&1&$nm$&$\max\{n,m\}$&id\\
star-center&1&complete&1&1&$nm+\frac{m(m-1)}{2}$&$n+m-1$&$c_{i^n}$\\
complete&1&star-center&1&1&$nm+\frac{n(n-1)}{2}$&$n+m-1$&$c_{i^m}$\\
star-spoke&$n$&complete&1&$n$&$n+m-1+\frac{m(m-1)}{2}$&$n+m-1$&$c_{i^n}\circ c_{i^m}$\\
complete&1&star-spoke&$m$&$m$&$n+m-1+\frac{n(n-1)}{2}$&$n+m-1$&$c_{i^m}\circ c_{i^n}$\\
star-spoke&$n$&star-spoke&$m$&$nm$&$n+m-1$&$\max\{n,m\}$&$c_{i^n}\circ c_{i^m}\circ c_{i^n}$\\
\hline
\end{tabular}
}
\caption{
The QASST-symmetry classes for the LC orbit of a complete bipartite graph $K_{n,m}$, where $n,m\geq2$. All graphs locally equivalent to $K_{n,m}$ have a split decomposition with two quotient graphs $Q_1$ and $Q_2$ falling into one of the six cases enumerated here, based on the derivations in~\cite{connolly2026local}.
Each quotient graph $Q_1$ and $Q_2$ is either complete, star-center (the single split-node is the center of the star), or star-spoke (the split-node is a spoke of the star).
The total number of graphs in the LC orbit is $nm+n+m+3$.
Reading the table, we see that both the number of edges and $\Delta(G)$ is minimized in the sixth case, occurring when $G$ is a \textit{binary-star}.
The last column includes an explicit formula using local complements that transforms $K_{n,m}$ into a graph of the given symmetry class. Here, the index notation means $i^n\in\{1,\cdots,n\}$ and $i^m\in\{1+n,\cdots,m+n\}$, referring to leaf-nodes in $Q_1$ and $Q_2$, respectively.
}
\label{tab:complete_bipartite_LC_orbit}
\end{table*}

\begin{table*}
\begin{center}
\begin{tabular}{|c|cc|c|c|c|}
\hline
$|V(k_{n,m})|$&$n$&$m$&$|{\mathcal O}(K_{n,m})|$&$\displaystyle{\min_{G\in{\mathcal O}(K_{n,m})}|E(G)|}$&$\displaystyle{\min_{G\in{\mathcal O}(K_{n,m})}\Delta(G)}$\\
\hline
4&2&2&11&3&2\\
\hline
5&3&2&14&4&3\\
\hline
6&4&2&17&5&4\\
&3&3&18&5&3\\
\hline
7&5&2&20&6&5\\
&4&3&22&6&4\\
\hline
8&6&2&23&7&6\\
&5&3&26&7&5\\
&4&4&27&7&4\\
\hline
9&7&2&26&8&7\\
&6&3&30&8&6\\
&5&4&32&8&5\\
\hline
10&8&2&29&9&8\\
&7&3&34&9&7\\
&6&4&37&9&6\\
&5&5&38&9&5\\
\hline
11&9&2&32&10&9\\
&8&3&38&10&8\\
&7&4&42&10&7\\
&6&5&44&10&6\\
\hline
12&10&2&35&11&10\\
&9&3&42&11&9\\
&8&4&47&11&8\\
&7&5&50&11&7\\
&6&6&51&11&6\\
\hline
\end{tabular}
\end{center}
\caption{
Summary of parameters and optimizations for graphs locally equivalent to a complete bipartite graph $K_{n,m}$ with up to 12 vertices. The size of the LC orbit is always $|{\mathcal O}(K_{n,m})|=nm+n+m+3$. Across the LC orbit, the minimal edge representative has $|E(G)|=n+m-1$ edges, and the minimal maximum vertex degree is $\Delta(G)=\max\{n,m\}$.
These both occur when $G\in{\mathcal O}(K_{n,m})$ is a \textit{binary-star}.
}
\label{tab:bipartite_fixed_parameters}
\end{table*}

Complete multipartite graphs and clique-stars are more advanced families of DH graph which we have also characterized using this technique. However, the combinatorial arguments are quite involved, and so we refer to the reader to the appendix of~\cite{connolly2026local} for the technical details.
Among these results, we identify minimal edge and minimal $\Delta(G)$ representatives of the LC orbits of $K_{n_1,\cdots,n_k}$ and $CS^r_{n_1,\cdots,n_k}$ based on the QASST structure of the graph, although unlike the complete bipartite case these do not always coincide.
Furthermore, these fall into a number of cases based on the size and parity of $k$ and the values of $n_1,\cdots,n_k$.
The formulas derived in~\cite{connolly2026local} are reproduced in Appendix~\ref{app:formulas_for_CM_and_CS}, but we do not explain all of the cases here.
To make these results more accessible, we include a Python script in the supplementary material that can be used to compute these formulas for any choice of input graph belonging to one of these two families~\cite{Connolly_Python}.
Table~\ref{tab:optimal_vs_split-fuse_comparison} in Section~\ref{sect:comparison_to_optimal_graph_states} also enumerates the optimized parameters for those graphs with up to 12 vertices according to these formulas.

\subsection{A Greedy Heuristic Algorithm}
\label{sect:heuristic}
Additionally, we propose a simple greedy heuristic algorithm applicable to generic graph states, distance-hereditary or otherwise.
This algorithm enumerates triangles within the graph, and applies local complement operations to remove an edge. Note that the computational complexity for enumerating triangles in a graph $G$ is $O(|E(G)|^{1.5})$~\cite{itai1978finding,chiba1985arboricity}, which is less costly than checking for the local equivalence of a pair of graphs. 
The resulting graph $H$ is updated if and only if the $c_v(H)$ has fewer edges, and the search for triangles continues, terminating when no triangles remain.
The total number of edges is reduced compared to the original graph state, but the algorithm is greedy, subject to becoming trapped in a local minimum, and depends on how triangles are enumerated.
However, unlike the other algorithms proposed in this paper, it requires no additional assumptions about the structure of the original graph state. Overall, the complexity of this heuristic in the worst case is $O(|E(G)|^{1.5} + |V(G)|^3)$.
Algorithm~\ref{alg:triangle_based_optimization} describes this technique in pseudocode. 

\begin{algorithm}[h]
    \SetKwInOut{Input}{Input}
    \SetKwInOut{Output}{Output}
    \Input{A graph $G$}
    \Output{Locally equivalent graph $H$, Sequence of local complement $S$, Vertex neighbors for the sequence of local complement $N$}
    
    let $H=G$\;
    Enumerate Triangles in $G$ and let $T=\{v|v \textrm{\,is in a triangle}\}$\;
    \tcc*[f]{$\mathcal{O}(|E(G)|^{1.5})$} 
    
    \ForAll{$v\in T$}{
        \tcc*[f]{$\mathcal{O}(|V(G)|^{3})$}\\
        $H'=c_v(H)$\;
        \eIf{$E(H')<E(H)$}
        {
        update $H \leftarrow c_v(H)$\;
        update $S \leftarrow S+c_v$\;
        update $N \leftarrow N+N(v)$
        }
        {Pass}
        {
        }
    }
    return $H,S,N$
    \caption{Triangle-based optimization for a graph state}
    \label{alg:triangle_based_optimization}
\end{algorithm}

\section{The Split-Fuse Method}
\label{sect:the_split_fuse_method}

Rather than merely searching the LC orbit for an optimal representative, we propose an alternative technique for preparing a DH graph state based on the QASST decomposition.
This method requires introducing some additional auxiliary qubits, but the total number of CZ gates and time steps both scale linearly with respect to the number of qubits in the target graph state.
Furthermore, it requires no additional information about the LC orbit.

We begin with a high-level overview of the proposed preparation protocol, followed by a detailed analysis of the resource requirements.
Then we compare the cost of using this technique versus using an optimal LC-equivalent graph state for the DH graph families we have classified.
Next, we discuss how this method can be generalized beyond the distance-hereditary case to prepare generic graph states.
Finally, we provide the results for numerical simulations comparing the effectiveness of the split-fuse method against methods searching the LC orbit.

\subsection{Preparation Protocol}

Rather than preparing a target graph state directly, the core idea of the split-fuse method is instead to prepare a number of smaller, intermediate graph states. These are then joined using Type-II fusions into the target graph state.
The intermediate graph states coincide with the quotient graphs from the split decomposition of the target state.
The QASST provides the recipe for this construction, with fusions corresponding to the edges in the underlying tree structure.
We state here as a proposition the fact that a graph state can always be reconstructed in this way, with a proof provided in Appendix~\ref{app:proof}.

\begin{proposition}\label{thm:reconstruction}
If $G$ is a graph with split decomposition into quotient graphs $\text{QASST}(G)=(Q_1,\cdots,Q_k)$, then the graph state $|G\rangle$ can be recovered from quotient graph states $|Q_1\rangle,\cdots,|Q_k\rangle$ by using Type-II fusions on all pairs of connected split-nodes described by the QASST.
\end{proposition}

Using this technique comes with some additional resource requirements, in particular a number of auxiliary qubits corresponding to the split-nodes in the quotient graphs.
These qubits are consumed in pairs by the Type-II fusions.
After assembling the target graph state, the remaining qubits correspond to the leaf-nodes from the quotient graphs.
However, this method has two principal advantages over direct implementation of a target graph state.
The first comes from the fusions themselves, which have the resource cost of a single CZ gate. Because a Type-II fusion introduces all-to-all connections between the fused qubits' neighbors, these achieve the same entanglement in one operation that would require many CZ gates to achieve directly.
Although two additional auxiliary qubits are required for each fusion, we save on potentially many costly CZ operations.

The second advantage of the split-fuse method is based on the fact that, before fusing, the intermediate graph states may be handled independently.
Preparing these states in parallel can potentially reduce the circuit depth, but more significantly allows for a reduction in CZ gates by exploiting LC equivalence.
The quotient graphs of the target graph state are determined by the split decomposition of the underlying graph, but these need not be implemented as intermediate states directly.
Rather, we are free to prepare an LC-equivalent graph state, and then use local Clifford operations to transform it into a quotient graph state before fusing.
In other words, we may further reduce the number of CZ gates by preparing an optimized graph state from the LC orbit of each quotient graph.

\begin{figure*}[t]
\centering
\small
\begin{tabular}{ccccc}
\includegraphics[scale=0.8,page=16]{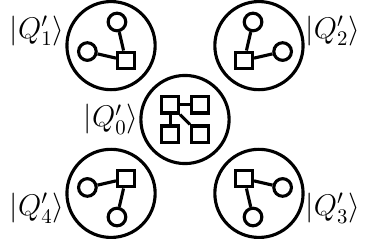}&&\includegraphics[scale=0.8,page=17]{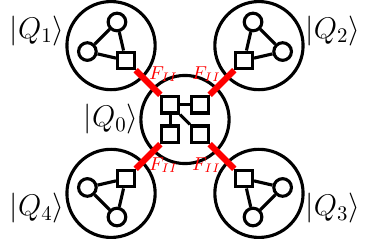}&&\includegraphics[scale=0.8,page=18]
{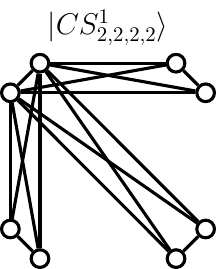}\\
&&&&\\
(a) Prepare $|Q_i'\rangle$ as stars&&(b) Convert to $|Q_i\rangle$ using LCs&&(c) Fuse split-node pairs
\end{tabular}
\normalsize
\caption{
An example of the split-fuse method for the preparation of the clique-star graph state $|CS^1_{2,2,2,2}\rangle$.
After identifying the quotient graphs $Q_0,\cdots,Q_4$ in the split decomposition of $|CS^1_{2,2,2,2}\rangle$, star graph states $|Q_0'\rangle\cdots|Q_4'\rangle$ with the same numbers of vertices are prepared. Each of these is converted into a corresponding quotient graph state $|Q_i\rangle$ using local Clifford operations.
Finally, Type-II fusions are used on pairs of split-nodes based on the red edges in the QASST to join these into the target graph state $|CS^1_{2,2,2,2}\rangle$. In total, this example uses 16 qubits (8 auxiliary), 15 CZ gates (4 fusions), and 5 time steps.
The corresponding preparation circuit for $|CS^1_{2,2,2,2}\rangle$ is provided in Figure~\ref{fig:split_fuse_circuit}.
}
\label{fig:split_fuse_example}
\end{figure*}

This is particularly advantageous for distance-hereditary graph states, for which all quotient graphs are either star or complete.
Since star and complete graphs are locally equivalent, we may initialize each quotient graph state of a target DH graph state as a star graph, and the transform some subset of these into complete graphs using local Clifford operations before fusing.
Star graphs have a minimal number of edges across the LC orbit, and hence this approach guarantees a minimal number of CZ operations before fusing.
Figure~\ref{fig:split_fuse_example} shows an example of this protocol for the preparation of a clique-star graph state.

Figure~\ref{fig:split_fuse_circuit} shows the split-fuse circuit used to construct the graph state illustrated in Figure~\ref{fig:split_fuse_example}.
The quotient graph states are disjoint and may be prepared simultaneously; these correspond to the portions of the circuit highlighted in blue.
After these are prepared, all Type-II fusions are performed at the same time and the target graph state is obtained.

\begin{figure}[htbp]
\includegraphics[width=\linewidth,page=24]{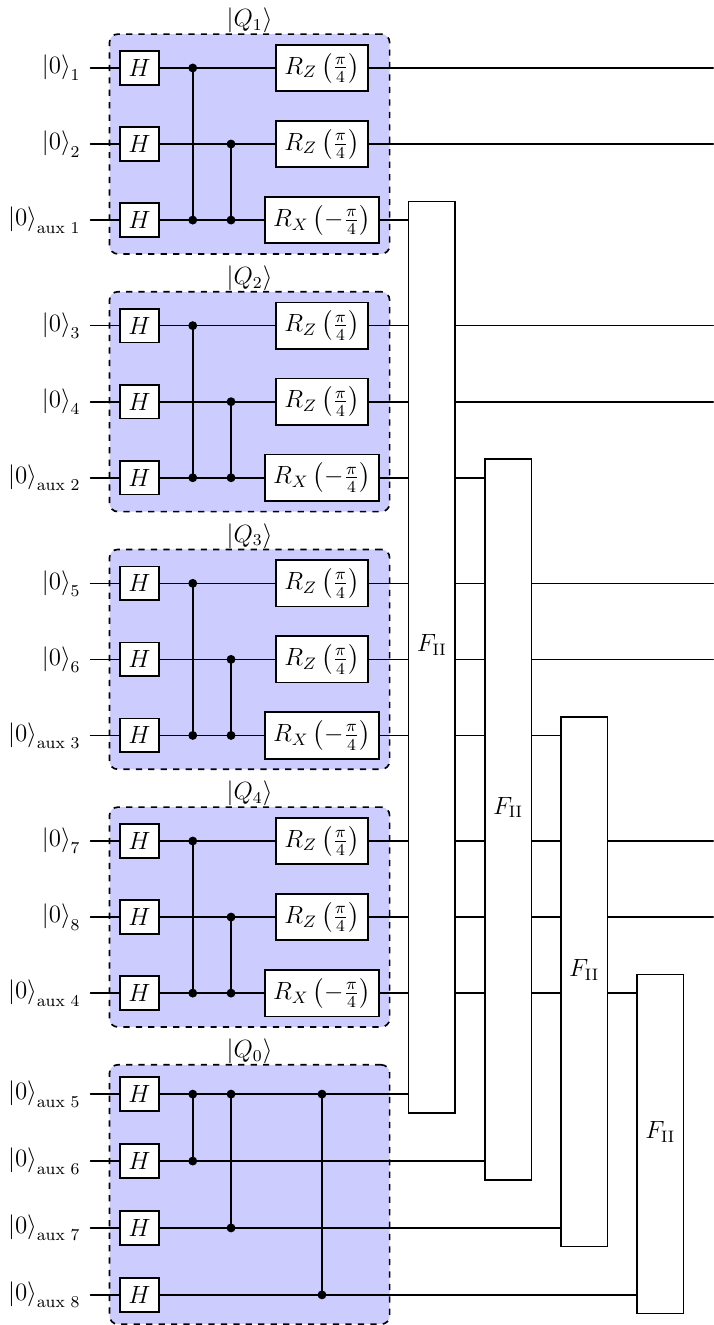}
\caption{
The circuit for assembling the clique-star graph state $|CS^1_{2,2,2,2}\rangle$ using the split-fuse method, based on the example of Figure~\ref{fig:split_fuse_example}.
Here, $F_{\text{II}}$ denotes a Type-II fusion between auxiliary qubits which are consumed in the process. All fusions are performed simultaneously. 
}
\label{fig:split_fuse_circuit}
\end{figure}

\subsection{Resource Cost for DH Graph States}

Let $G$ be any distance-hereditary graph with $n$ vertices and a split decomposition into $k$ quotient graphs $\textit{QASST}(G)=(Q_1,\cdots,Q_k)$. Let $n_i$ denote the number of leaf nodes in quotient graph $Q_i$, so that $n=n_1+\cdots+n_k$. Let $d_i$ denote the number of split-nodes in quotient graph $Q_i$, noting that $d_i=\text{deg}(Q_i)$ matches the degree of the quotient graph as a vertex in the graph-labeled strong split tree. Let $r$ denote the number of edges in this graph-labeled tree and let $d=\sum_{i=1}^kd_i=2r$ be the total degree.
Since the QASST is an acyclic graph with $k$ vertices and $r$ edges, note that $r=k-1$.

Our goal is to prepare a target DH graph state $|G\rangle$ with $n$ qubits.
We require an additional $d=2r=2k-2$ auxiliary qubits to be consumed by fusions.
Since each quotient graph of $G$ is locally equivalent to a star, we prepare $k$ intermediate star graph states $|Q_1'\rangle,\cdots,|Q_k'\rangle$. Each of these $|Q_i'\rangle$ contains $n_i+d_i$ qubits, matching $n_i$ leaf-nodes and $d_i$ split-nodes in the corresponding quotient graph $Q_i$.
We then use local Clifford operations as needed to change these into the quotient graph states $|Q_1\rangle,\cdots,|Q_k\rangle$.
Finally, we perform $r=k-1$ Type-II fusions on pairs of qubits matching split-nodes to assemble the target graph state $|G\rangle$.

We summarize the resource requirements of this strategy as follows.
Constructing each quotient graph state $|Q_i\rangle$ requires $n_i+d_i-1$ CZ gates, followed possibly by a local Clifford operation.
In total, this amounts to $\sum_{i=1}^k(n_i+d_i-1)=n+d-k=n+2r-k$ CZ gates across all quotient graph states, plus an additional $r=k-1$ Type-II fusions.
Each star graph state requires $n_i+d_i-1$ time steps in the preparation circuit, plus an additional time step if a local Clifford operation is needed.
These may be prepared in parallel, and so the total required time steps depends only on the largest quotient graph.
All fusions may be performed simultaneously, contributing 1 to the number of time steps, and consuming the $d$ auxiliary qubits.

This leaves a single connected $n$-qubit graph state corresponding to $|G\rangle$.
Counting the Type-II fusions as CZ gates, we summarize the preparation resources used in this process as follows.
\begin{eqnarray}
\text{total CZ gates}&=&n+2k-3\label{eq:split_fuse_CZ}\\
\text{total time steps}&=&1+\max_{i=1}^k\{n_i+d_i\}\label{eq:split_fuse_time_steps}\\
\text{total qubits}&=&n+2k-2\label{eq:split_fuse_qubits}
\end{eqnarray}
Notably, since the parameters $k,d,r\leq n$, the three resources described above all scale linearly with the number of qubits $n$ in the target graph state.
An example of the three steps in this process is outlined in Figure~\ref{fig:split_fuse_example}. 
This can be implemented for any DH graph state provided its split decomposition as described by the QASST is known.

For general simple undirected graphs, DH or otherwise, there exist efficient algorithms for computing the split decomposition~\cite{ma1994n2,dahlhaus1994efficient,dahlhaus2000parallel,charbit2012linear}. The fastest of these is linear-time with respect to the number of vertices~\cite{charbit2012linear}.

\subsection{Comparison to Optimal Graph States}
\label{sect:comparison_to_optimal_graph_states}

In the context of optimizing preparation resources, it is natural to ask whether using the split-fuse method to prepare a target graph state provides any advantages over direct implementation of an optimal LC-equivalent graph state followed by local Clifford transformations to the target graph state.
In the case of graph states LC-equivalent to $|K_{n,m}\rangle$, the answer is no.
As we observed in Section~\ref{sect:optimizations_across_the_LC_orbit}, the optimal representative of the LC orbit of a complete bipartite graph is a binary-star, which is acyclic. In particular, this means that direct implementation of this optimal graph state already requires a minimal number of edges and time steps, and no auxiliary qubits.
However, this property is unique to acyclic graphs, and generic graph states will not contain such a representative in the LC orbit.

The split-fuse method can provide an advantage for preparing graph states locally equivalent to a complete multipartite graph or a clique-star, although the improvements only become apparent for larger numbers of qubits.
Table~\ref{tab:optimal_vs_split-fuse_comparison} shows a comparison of the resources required for this technique versus those required for direct implementation of an optimal graph state in these two cases.
We show all such graph states with up 12 qubits; in these small cases, direct implementation is always more efficient with regards to the number of CZ gates, although in some of these cases the split-fuse method requires slightly fewer time steps.
When we reach 24 qubits, we encounter the first graph state for which this technique requires both fewer CZ gates and fewer time steps. We show a selected subset of such graph states in the table, showing that the number of cases that gain an advantage from the split-fuse method increases with the number of qubits.
These numerical comparisons were generated in Python code available on GitHub~\cite{Connolly_Python}, and one may verify that this pattern continues for higher numbers of qubits.

\begin{table*}
\begin{center}
\scalebox{0.64}{
\begin{tabular}{|cc|cccccc|ccc|ccc|ccc|}
\hline
&&\multicolumn{6}{c|}{$n=n_1+\cdots+n_k=|V(G)|$}&\multicolumn{3}{c|}{$G\cong_{\text{LC}}K_{n_1,\cdots,n_k}$}&\multicolumn{3}{c|}{$G\cong_{\text{LC}}CS^r_{n_1,\cdots,n_k}$}&\multicolumn{3}{c|}{split-fuse method}\\
$n$&$k$&$n_1$&$n_2$&$n_3$&$n_4$&$n_5$&$n_6$&$|{\mathcal O}(G)|$&$\min |E(G)|$&$\min\Delta(G)+1$&$|{\mathcal O}(G)|$&$\min |E(G)|$&$\min\Delta(G)+1$&CZ gates&time steps&qubits (aux.)\\
\hline
6&3&2&2&2&&&&40&6&4&41&6&4&11&4&12 (6)\\
\hline
7&3&3&2&2&&&&50&7&5&52&7&5&12&5&13 (6)\\
\hline
8&3&4&2&2&&&&60&8&6&63&8&6&13&6&14 (6)\\
&3&3&3&2&&&&62&8&5&66&8&5&13&5&14 (6)\\
&4&2&2&2&2&&&149&10&5&148&9&4&15&5&16 (8)\\
\hline
9&3&5&2&2&&&&70&9&7&74&9&7&14&7&15 (6)\\
&3&4&3&2&&&&74&9&6&80&9&6&14&6&15 (6)\\
&3&3&3&3&&&&76&10&{\color{red}6}&84&9&5&14&5&15 (6)\\
&4&3&2&2&2&&&190&11&5&188&10&5&16&5& 17 (8)\\
\hline
10&3&6&2&2&&&&80&10&8&85&10&8&15&8&16 (6)\\
&3&5&3&2&&&&86&10&7&94&10&7&15&7&16 (6)\\
&3&4&4&2&&&&88&10&6&97&10&6&15&6&16 (6)\\
&3&4&3&3&&&&90&11&{\color{red}7}&102&10&6&15&6&16 (6)\\
&4&4&2&2&2&&&231&12&6&228&11&6&17&6&18 (8)\\
&4&3&3&2&2&&&242&12&5&238&11&5&17&5&18 (8)\\
&5&2&2&2&2&2&&526&12&5&527&13&6&19&6&20 (10)\\
\hline
11&3&7&2&2&&&&90&11&9&96&11&9&16&9&17 (6)\\
&3&6&3&2&&&&98&11&8&108&11&8&16&8&17 (6)\\
&3&5&4&2&&&&102&11&7&114&11&7&16&7&17 (6)\\
&3&5&3&3&&&&104&12&{\color{red}8}&120&11&7&16&7&17 (6)\\
&4&5&2&2&2&&&272&13&7&268&12&7&18&7&19 (8)\\
&3&4&4&3&&&&106&12&{\color{red}7}&124&11&6&16&6&17 (6)\\
&4&4&3&2&2&&&294&13&6&288&12&6&18&6&19 (8)\\
&4&3&3&3&2&&&308&13&5&300&12&5&18&5&19 (8)\\
&5&3&2&2&2&2&&674&13&5&676&14&6&20&6&21 (10)\\
\hline
12&3&8&2&2&&&&100&12&10&107&12&10&17&10&18 (6)\\
&3&7&3&2&&&&110&12&9&122&12&9&17&9&18 (6)\\
&3&6&4&2&&&&116&12&8&131&12&8&17&8&18 (6)\\
&3&6&3&3&&&&118&13&{\color{red}9}&138&12&8&17&8&18 (6)\\
&4&6&2&2&2&&&313&14&8&308&13&8&19&8&20 (8)\\
&3&5&5&2&&&&118&12&7&134&12&7&17&7&18 (6)\\
&3&5&4&3&&&&122&13&{\color{red}8}&146&12&7&17&7&18 (6)\\
&4&5&3&2&2&&&346&14&7&338&13&7&19&7&20 (8)\\
&3&4&4&4&&&&124&14&{\color{red}8}&151&12&6&17&6&18 (6)\\
&4&4&4&2&2&&&357&14&6&348&13&6&19&6&20 (8)\\
&4&4&3&3&2&&&374&14&6&362&13&6&19&6&20 (8)\\
&5&4&2&2&2&2&&822&14&6&825&15&6&21&6&22 (10)\\
&4&3&3&3&3&&&392&14&{\color{red}6}&376&15&{\color{red}6}&19&5&20 (8)\\
&5&3&3&2&2&2&&862&14&5&866&15&6&21&6&22 (10)\\
&6&2&2&2&2&2&2&1823&16&7&1822&15&6&23&7&24 (12)\\
\hline
\multicolumn{16}{c}{\vdots}\\
\hline
24&3&8&8&8&&&&436&{\color{red}30}&{\color{red}16}&779&24&10&29&10&30 (6)\\
&4&6&6&6&6&&&2885&26&{\color{red}9}&2260&{\color{red}33}&{\color{red}12}&31&8&32 (8)\\
\hline
25&3&9&8&8&&&&470&{\color{red}31}&{\color{red}17}&862&25&11&30&11&31 (6)\\
&4&7&6&6&6&&&3266&27&{\color{red}10}&2516&{\color{red}34}&{\color{red}13}&32&9&33 (8)\\
&5&5&5&5&5&5&&9856&{\color{red}36}&{\color{red}10}&10880&30&{\color{red}9}&34&7&35 (10)\\
\hline
26&3&10&8&8&&&&504&{\color{red}32}&{\color{red}18}&945&26&12&31&12&32 (6)\\
&3&9&9&8&&&&506&{\color{red}32}&{\color{red}17}&954&26&11&31&11&32 (6)\\
&4&8&6&6&6&&&3647&28&{\color{red}11}&2772&{\color{red}35}&{\color{red}14}&33&10&34 (8)\\
&4&7&7&6&6&&&3698&28&{\color{red}10}&2798&{\color{red}35}&{\color{red}13}&33&9&34 (8)\\
&5&6&5&5&5&5&&11240&{\color{red}37}&{\color{red}11}&12520&31&{\color{red}10}&35&8&36 (10)\\
\hline
27&3&11&8&8&&&&538&{\color{red}33}&{\color{red}19}&1028&27&13&32&13&33 (6)\\
&3&10&9&8&&&&542&{\color{red}33}&{\color{red}18}&1046&27&12&32&12&33 (6)\\
&3&9&9&9&&&&544&{\color{red}34}&{\color{red}18}&1056&27&11&32&11&33 (6)\\
&4&9&6&6&6&&&4028&29&{\color{red}12}&3028&{\color{red}36}&{\color{red}15}&34&11&35 (8)\\
&4&8&7&6&6&&&4130&29&{\color{red}11}&3080&{\color{red}36}&{\color{red}14}&34&10&35 (8)\\
&4&7&7&7&6&&&4188&29&{\color{red}10}&3108&{\color{red}36}&{\color{red}13}&34&9&35 (8)\\
&5&7&5&5&5&5&&12624&{\color{red}38}&{\color{red}12}&14160&32&{\color{red}11}&36&9&37 (10)\\
&5&6&6&5&5&5&&12808&{\color{red}38}&{\color{red}11}&14408&32&{\color{red}10}&36&8&37 (10)\\
\hline
\end{tabular}
}
\end{center}
\caption{
Comparison of optimal graph states from the LC orbits of $K_{n_1,\cdots,n_k}$ and $CS^r_{n_1,\cdots,n_k}$ with the resource requirements of the split-fuse method, where $n=n_1+\cdots+n_k$ is the number of qubits in the target state (auxiliary qubits consumed by fusions are indicated in parentheses). 
Any choice of graph state $|G\rangle$ requires $|E(G)|$ CZ gates and at most $\Delta(G)+1$ time steps to implement directly, whereas the split-fuse method requires $n+2(k+1)-3$ CZ gates and $2+\max\{k-1,n_1,\cdots,n_k\}$ time steps for the graphs in this table.
As the number of qubits increases, the resource gap between using an optimal state and the split-fuse method decreases, with the first improvement in the number of time steps occurring when $n=9$ and the first improvement in the number of CZ gates occurring when $n=24$.
All graph states in these families with up to 12 qubits are displayed, along with a selected subset of larger states benefiting from the split-fuse method.
Entries are highlighted in red when the split-fuse method surpasses the optimal state.
The number of graph states which benefit from using the split-fuse method increases with $n$.
Whenever the number of CZ gates is improved, we also observe an improvement in the number of time steps.
}
\label{tab:optimal_vs_split-fuse_comparison}
\end{table*}

Although a choice of optimal graph state from the LC orbit is often superior for small cases, the true advantage of the split-fuse method is the simplicity of identifying the preparation protocol. This comes immediately with knowledge of the split decomposition, for which efficient algorithms are known, and requires no additonal knowledge about locally equivalent graphs.
The complete classification of the LC orbit we derive in~\cite{connolly2026local} for complete multipartite graphs and clique-stars is quite technical, but we see from the table that the split-fuse method is almost as good even for small graph states.
However, the split-fuse method may be applied for DH graphs whose LC orbit has not been classified and an optimal graph representative is not known. This is especially promising given how quickly the size of the LC orbit grows.
Even without a comparison for other graph states, given how the resource requirements scale linearly, we expect that the split-fuse method will surpass using an optimal graph state for general DH graph states that are sufficiently large.

\subsection{Generalization to non-DH Graph States}

We initially define the split-fuse method for distance-hereditary graph states, but there is a natural generalization to generic simple graphs.
The split-fuse preparation protocol is based on the quotient graphs in the split decomposition, which exists for any graph. The main advantage comes from preparing optimal intermediate graph states LC-equivalent to the quotient graph states.
When restricting to DH graphs in particular, these will always be star graph states, which allowed us to derive the simple formulas of Equations~\ref{eq:split_fuse_CZ}, \ref{eq:split_fuse_time_steps}, and \ref{eq:split_fuse_qubits}.
However, the same concept can be applied to non-DH graphs more generally.

The split decomposition of a non-DH graph will contain at least one prime quotient graph.
Unfortunately, there are many different types of prime graphs, making a general characterization with simple formulas unlikely.
Even so, optimal prime quotient graphs could be identified in specific cases by exploring their LC orbits, as has already been done numerically for small graph states with up to 12 qubits~\cite{cabello2011optimal,adcock2020mapping}.
In those cases where optimal prime quotient graphs are already known, these could be used with local Clifford operations to prepare the quotient graph states of a corresponding non-DH graph, followed by fusions to assemble the target non-DH graph state.

This suggests a hybrid split-fuse strategy for preparing non-DH graph states.
This would begin by computing the QASST of a target graph state via the split decomposition.
For all quotient graphs with known optimal representatives (including star and complete quotient graphs), these could be prepared and converted into quotient graph states.
The remaining quotient graph states could be prepared directly before fusing into the target graph state.
Based on this idea, future work adapting the split-fuse method to this hybrid strategy would begin with an examination of optimal prime graphs with as any many qubits as feasible.

\subsection{Numerical Comparison of Techniques}
\label{sect:numerical_comparison}

\begin{figure*}[t]
\centering

\includegraphics[width=0.95\linewidth]{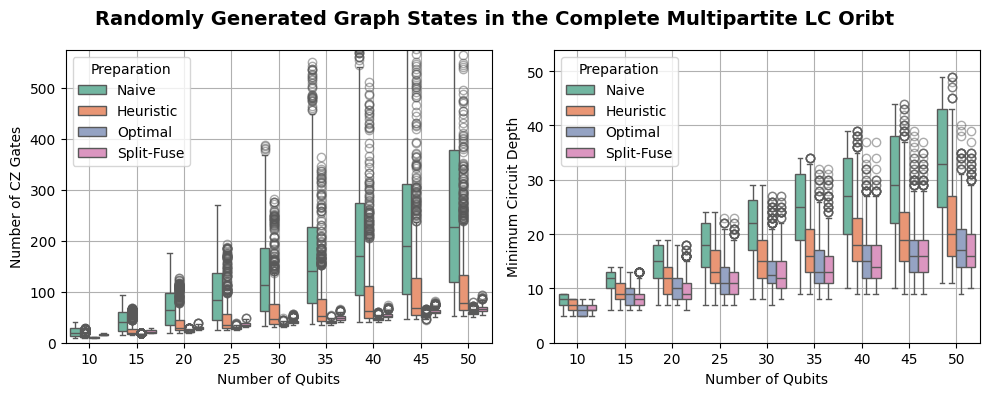}

\includegraphics[width=0.95\linewidth]{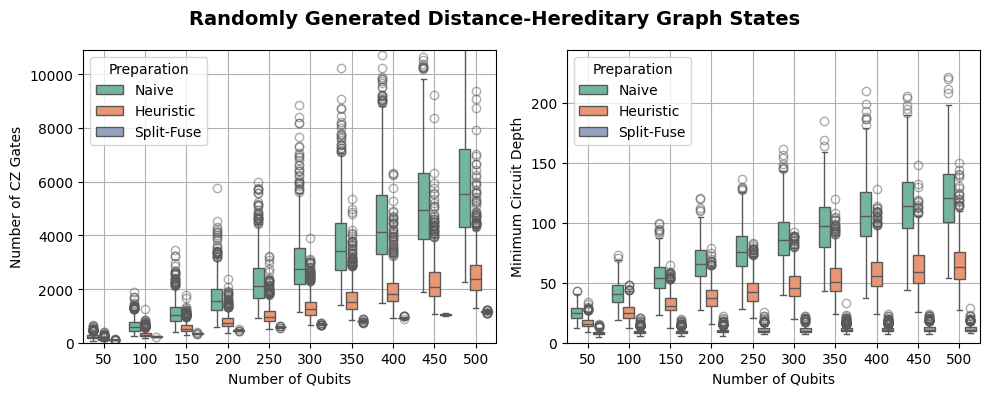}
\caption{
Comparison of the resource requirements (CZ gates on the left and circuit depth on the right) for the various graph state preparation protocols.
Boxplots use 1000 randomly generated sample graphs each.
These use graphs from the LC orbit of a complete multipartite graph (above), or generic distance-hereditary graphs (below).
The naive method serves as the baseline, in which a randomly generated graph state is prepared solely using CZ gates.
This is compared against the direct preparation of the heuristic-improved version of this graph, as well as the optimal locally equivalent graph, where it is known for the complete multipartite case.
We also compare these against the resources required for building the original randomly generated graph state using the split-fuse technique.
}
\label{fig:boxplot_comparisons}
\end{figure*}

In addition to the data summarized in Table~\ref{tab:optimal_vs_split-fuse_comparison}, Figure~\ref{fig:boxplot_comparisons} also shows numerical simulations comparing various preparation methods for randomly generated graph states of different sizes.
In the case of graph states locally equivalent to a complete multipartite graph, where the optimal representative of the LC orbit is known, we see that the greedy heuristic algorithm often comes close to this bound even when a heuristic-improved graph is not optimal.
By contrast, the resources required by the split-fuse method are almost always comparable to and occasionally superior to direct implementation for an optimal graph state, consistent with the data in Table~\ref{tab:optimal_vs_split-fuse_comparison}.
These facts hold both for the number of CZ gates and the minimum circuit depth.

The numerical simulations for randomly generated distance-hereditary graphs show a similar relationship between the preparation techniques, but the optimal representative of the LC orbit is not known in these cases.
Even so, we see that the resources required by the heuristic-improved graph state as well as the split-fuse technique are a fraction of those required by naive implementation.
The split-fuse method in particular scales exceptionally well compared with the original randomly generated graph. This is especially true for the circuit depth, which appears almost constant.

This phenomenon can be explained by the fact that the number of time steps required by the split-fuse method is determined by the largest quotient graph in the split decomposition, but is unaffected by the total number of quotient graphs.
The random distance-hereditary graphs we simulate are constructed recursively through a process of randomly selected one-vertex extensions by twins~\cite{bandelt1986distance}. In terms of the QASST, a one-vertex extension either introduces a new quotient graph, or increases the size of an existing quotient graph by one.
It is more likely that the QASST of a graph randomly generated in this way has many small quotient graphs of similar sizes on average than a few large quotient graphs, which explains why the number of required time steps (and hence circuit depth) grows so slowly in this case.

\begin{figure*}[htbp]
\centering

\includegraphics[width=0.95\linewidth]{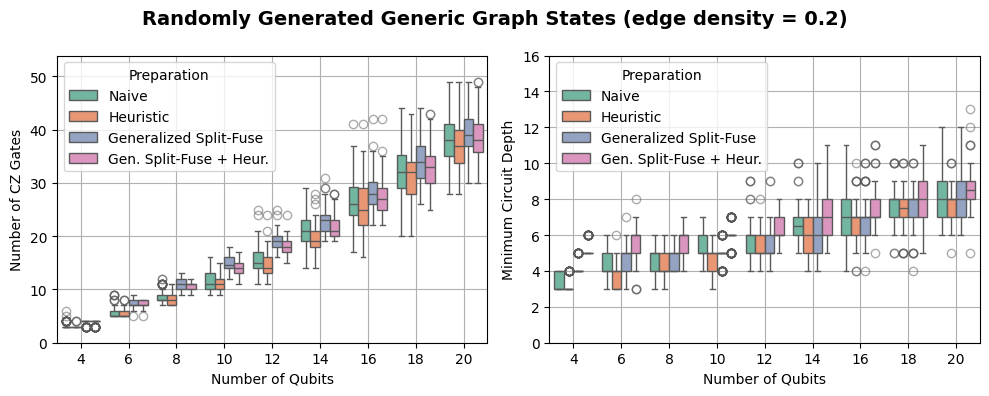}

\includegraphics[width=0.95\linewidth]{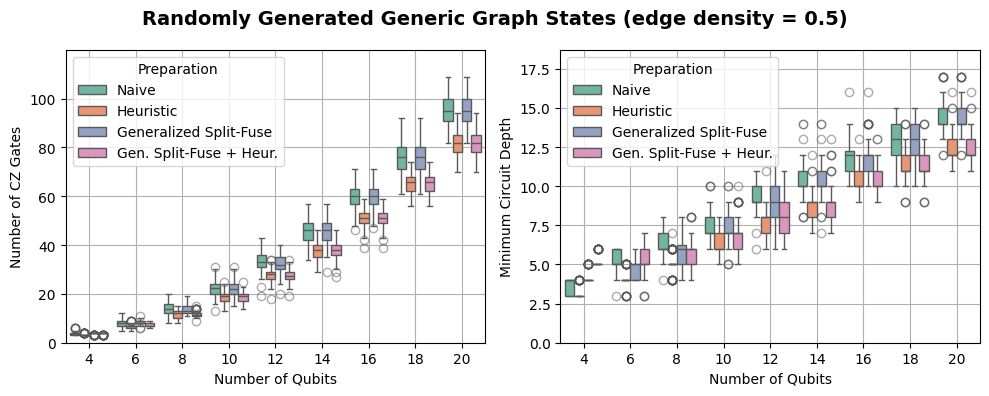}

\includegraphics[width=0.95\linewidth]{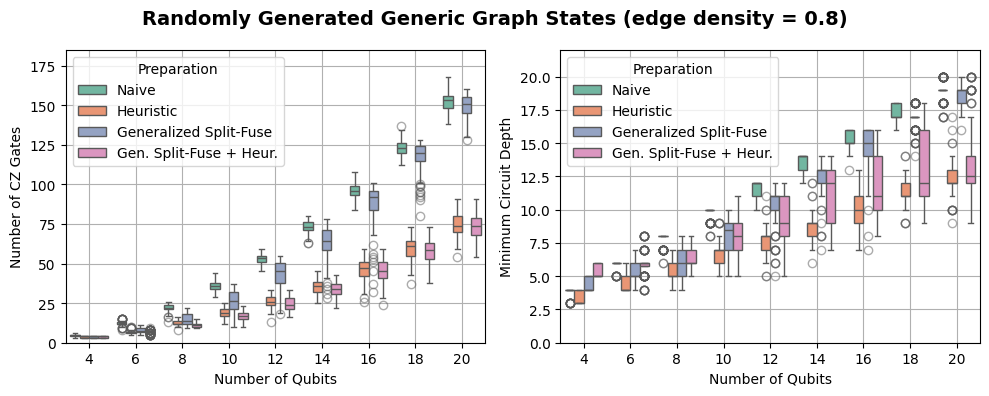}

\caption{
Comparison of the resource requirements (CZ gates on the left and circuit depth on the right) for the various preparation methods for generic randomly generated connected Erdős–Rényi graph states with different edge densities.
The baseline naive preparation and heuristic-improved preparation are the same methods used in Figure~\ref{fig:boxplot_comparisons}. The generalized split-fuse method relies on computing the QASST of the graph; star and complete quotient graphs are prepared as stars (which require minimal resources), while prime quotient graphs are prepared naively.
The combined generalized split-fuse + heuristic method replaces prime quotient graphs with heuristic-improved versions.
Computation time is dominated by the subroutine to compute the QASST, and so only relatively small graph states are compared, with boxplots using 100 randomly generated sample graphs each.
}
\label{fig:random_graph_boxplot_comparisons}
\end{figure*}

Figure~\ref{fig:random_graph_boxplot_comparisons} shows a comparison of the preparation protocols for generic randomly generated connected Erdős–Rényi graph states with various edge densities.
Here, we use two versions of a generalized split-fuse method for non-DH graphs: a basic version, wherein prime quotient graphs are prepared naively, and a version combined with the heuristic algorithm to improve prime quotient graphs.
These are compared to naive preparation and the heuristic algorithm applied to the original graph.
There are distinct differences in behavior using random graphs with different edge densities.

At low edge densities, the generalized split-fuse method actually degrades the resource requirements.
This is consistent with the behavior we see in Table~\ref{tab:optimal_vs_split-fuse_comparison} for small graph states, and can be explained by the fact that the quotient graphs are relatively small. The additional resources introduced by auxiliary qubits offset the benefits of preparing an optimized quotient graph when the quotients themselves are small.
Furthermore, the low edge density limits the improvements obtainable by the proposed preparation protocols.
For larger or denser graph states, we would expect to see clearer improvements.

At middle edge densities, we see another interesting phenomenon.
The naive and generalized split-fuse methods are identical, as are the heuristic and combined methods.
This can only occur when the randomly generated graph is itself prime, which implies that the QASST consists of a single quotient graph.
With high probability, this suggests that randomly generated graphs with around 0.5 edge density are prime.
In these cases, the split-fuse method provides no benefit, although the heuristic algorithm can still be useful.

However, at high edge densities, we begin to see the expected improvements.
The generalized split-fuse method provides modest improvements, suggesting the existence of multiple complete quotient graphs in the QASST, while the best improvements come from the combined protocol.
Inspections of the QASSTs for the randomly generated graphs usually reveal the existence of one large prime quotient graph and several small complete quotient graphs, explaining why the heuristic algorithm seems to have such a disproportionate effect.
As in the other cases, we expect the improvements to become more pronounced as the graph states become larger.
Furthermore, we can see that these techniques are most beneficial to graph states with high edge density.

\section{Discussion and Conclusions}
\label{sect:conclusion}

In this work, we have investigated the optimization of graph state preparation within local Clifford (LC) equivalence classes, focusing on the broad and structurally rich class of distance-hereditary (DH) graphs. By exploiting the closure of DH graphs under local complement operations and their canonical split decomposition, we developed scalable methods that avoid exhaustive enumeration of LC orbits, which becomes infeasible for large systems.

Our first contribution is an analytic classification of LC orbits for several families of DH graphs of practical relevance, including complete bipartite and multipartite graphs, clique-stars, and multi-leaf repeater graphs. Using the quotient-augmented strong split tree (QASST) representation, we identified symmetry classes within each orbit and determined LC-equivalent representatives that minimize either the number of entangling gates or the maximum vertex degree, and hence the minimum achievable preparation depth. These results provide provably optimal constructions for these families and extend existing numerical studies to substantially larger graph sizes.

Our second and more general contribution is the split–fuse preparation method for distance-hereditary graph states along with its generalized version to generic graph states. This construction leverages the split decomposition to assemble a target graph state from smaller, efficiently preparable components using fusion operations. For DH graphs, all quotient components can be chosen as stars, which are locally equivalent to complete graphs and require minimal entangling resources. As a result, in the distance-hereditary case, the split–fuse method achieves linear scaling in the number of controlled-Z gates, time steps, and auxiliary qubits, independent of the size of the LC orbit. Numerical comparisons show that while direct preparation of optimal LC representatives can be more efficient for small systems, the split–fuse method becomes advantageous as the graph size increases.
By contrast, the generalized version the split-fuse method shows promise when applied to dense graphs, especially when combined with other algorithms capable of simplifying prime quotient graphs.

The split–fuse approach introduces additional auxiliary qubits and relies on fusion operations, whose success probabilities and noise characteristics depend on the physical platform. These factors represent an important practical consideration and may offset some of the asymptotic advantages in near-term implementations. Nevertheless, the linear scaling of entangling resources and circuit depth makes the method particularly attractive for large-scale graph state generation, especially in photonic and modular architectures where fusion-based operations are natural.

Although our analysis focuses on distance-hereditary graphs, the underlying ideas extend more broadly. General graphs admit split decompositions that include prime components, which cannot be reduced to stars. In such cases, one may combine the split–fuse framework with optimized preparation strategies for prime graph states, potentially yielding hybrid constructions that retain favorable scaling properties.
As shown in our simulations, the generalized version of the split-fuse technique provides an improvement for the preparation of dense generic graph states, especially when combined with the heuristic algorithm.
Exploring these generalizations further, as well as incorporating realistic noise models and hardware constraints, are promising directions for future work.

Overall, our results demonstrate that combining structural graph-theoretic insights with operational considerations enables scalable optimization of graph state preparation beyond what is accessible through brute-force LC orbit searches. We expect these methods to be useful for the design of large entangled resource states in measurement-based quantum computation, quantum networking, and related applications.

\section*{Data Availability}
All programs and data used in this paper are available in a GitHub repository\cite{Connolly_Python}.
This includes a simple Python notebook tutorial which can be used to reproduce our numerical results.

\section*{Acknowledgments}
\label{sect:acknowledgements}

NC would like to thank Kenneth Goodenough for his significant help in building an intuition for working with graph states.
This work was supported in part by the JST Moonshot R\&D Program (Grant Nos. JPMJMS2061 \& JPMJMS226C).
SN has been supported by the New Energy and Industrial Technology Development Organization (NEDO) JPNP23003, JSPS Overseas Challenge Program for Young Researchers, and the JSPS Overseas Research Fellowships.

\bibliographystyle{plain}
\bibliography{references.bib}

\appendix
\section{Graph Theory Basics}
\label{app:graph_basic}
A \textit{graph} $G=(V,E)$ consists of a set of vertices $V$ and edges $E\subseteq V\times V$ between pairs of vertices. In this paper, we assume that graphs are \textit{simple} (contain no multi-edges nor self-loops), \textit{undirected} $((u,v)=(v,u)\in E)$, \textit{connected} (there exists a path in $G$ between any pair of vertices), and \textit{labeled} (vertices are labeled, usually by integers).
Given a graph $G$, we will denote its vertex set by $V(G)$ and its edge set by $E(G)$. Two vertices $u,v\in V(G)$ are called \textit{adjacent} if there exists an edge $e=(u,v)\in E(G)$, and $e$ is said to be \textit{incident} to $u$ and $v$.
The \textit{distance} between two vertices is the length (number of edges) of a minimal path between them.
Given a vertex $v\in V(G)$, the \textit{degree} of $v$ is number $\text{deg}(v)$ of edges in $E(G)$ incident to $v$, and the maximum vertex degree in the graph is denoted as $\Delta(G):=\max_{v\in V(G)}\text{deg}(v)$.
A vertex of degree 1 in $G$ is referred to as a \textit{leaf}.

Given graphs $H$ and $G$, $H$ is said to be a \textit{subgraph} of $G$ if $V(H)\subseteq V(G)$ and $E(H)\subseteq E(G)$. Given a subset of vertices $S\subseteq V(G)$, an \textit{induced subgraph} $G(S)$ of $G$ is the graph with vertex set $S$ and edge set consisting of all edges in $E(G)$ between any two vertices in $S$. An \textit{edge-induced subgraph} can be defined in a similar way.
For a graph $G$ with vertex $v\in V(G)$, the \textit{neighborhood} of $v$ is the set of all adjacent vertices: $N(v)=\{u:(u,v)\in E(G)\}$.
Let $G(v)$ denote the subgraph induced by $N(v)$.

In some cases involving local complements, it will be useful to use an operation known as an \textit{edge pivot}, which is a composition of three local complements defined by an edge in the graph. For a graph $G$ with an edge $e=(u,v)\in E(G)$, an edge pivot with respect to $(u,v)$ is defined by the formula
\begin{eqnarray}
\text{ep}_{(u,v)}&=&c_u\circ c_v\circ c_u.
\end{eqnarray}
Note that $\text{ep}_{(u,v)}(G)=\text{ep}_{(v,u)}(G)$; the transformation is the same when switching the roles of $u$ and $v$, and hence without loss of generality, this operation is defined by an undirected edge.
In the graph state case, an edge pivot defines the three corresponding local Clifford operations applied in sequence.

\section{Special Graph Classes}
\label{app:special_graph_classes}

A graph is a \textit{star} if it consists of a single central vertex connected to some number of leaves; $S_n$ denotes a star with $n$ leaves ("spokes" of the star).
A graph is \textit{complete} if it contains all possible edges between vertices; the complete graph on $n$ vertices is denoted $K_n$.
A graph $G$ is \textit{bipartite} if its vertex set can be partitioned into two sets $V(G)=V_1\sqcup V_2$ such that there do not exist edges between two vertices in the same set. A graph is called \textit{complete bipartite} if it is bipartite and there exist all possible edges between vertices in the sets $V_1$ and $V_2$. $K_{n,m}$ denotes the complete bipartite graph between disjoint sets of $n$ and $m$ vertices.
More generally, $G$ is called \textit{multipartite} or \textit{$k$-partite} if its vertex set can be partitioned into $V(G)=V_1\sqcup\cdots\sqcup V_k$ such that there exist no edges between vertices in the same $V_i$. Likewise, a graph is \textit{complete $k$-partite} if it is $k$-partite and contains all possible edges. We will use $K_{n_1,\cdots,n_k}$ to denote a complete $k$-partite graph, where $V(K_{n_1,\cdots,n_k})=V_1\sqcup\cdots\sqcup V_k$ with $|V_i|=n_i$.

Given a graph $G$, a \textit{clique} in $G$ is a complete subgraph.
A \textit{clique-star} is a graph which can be partitioned into $k$ cliques $H_1,\cdots,H_k$ with the property that there exists one central clique $H_r$ whose vertices are fully connected to the vertices in each of the other cliques $H_{i\neq r}$, but there exist no edges between the vertices from any other two cliques $H_{i\neq r}$ and $H_{j\neq r}$. A clique-star consisting of cliques $H_1,\cdots,H_k$ and with central clique $H_r$ is denoted $CS^r_{n_1,\cdots,n_k}$, where $n_i=|V(H_i)|$.

A \textit{repeater graph}~\cite{azuma2015all} $R_n$ consists of clique with $n$ vertices, each of which has a single leaf attached to it.
A \textit{multi-leaf repeater graph} $MR_{n_1,\cdots,n_k}$ is a generalization of a repeater consisting of a clique with $k$ vertices enumerated $1,\cdots,k$, where the $i^{\text{th}}$ vertex has $n_i-1$ leaves attached to it.

\section{Vertex Minors and Measurements}
\label{app:vertex_minors_and_measurements}

Closely related to graph local equivalence is the notion of vertex minors. Given graphs $H$ and $G$, $H$ is said to be a \textit{vertex minor} of $G$ if it can be obtained via some sequence of local complements and vertex deletions~\cite{oum2005rank}.
For graph states, a vertex deletion occurs whenever a qubit is measured, along with some local Clifford operation depending on the basis of the measurement.
Hence, any graph state obtainable by some combination of local Clifford operations and Pauli measurements will correspond to a vertex minor of the original graph state.
In general, the complexity of determining whether one graph is a vertex minor of another is NP-complete~\cite{dahlberg2022complexity}.

The effect of local Pauli measurements with respect to a choice of qubit in a graph state can also be be expressed in terms of local complements, with explicit formulas given in~\cite{hein2004multiparty,dahlberg2018transforming}.
For a graph state $|G\rangle$ and qubit corresponding to $v\in V(G)$, let $P_v^{(X,\pm)}$, $P_v^{(Z,\pm)}$, and $P_v^{(Y,\pm)}$ denote the $X$, $Y$, and $Z$ Pauli projectors on this qubit. A local Pauli measurement on $|G\rangle$ with respect to $v$ is then given by
\begin{eqnarray}
P_v^{(Z,\pm)}&=&\tfrac{1}{2}|Z,\pm\rangle_v\otimes U_v^{(Z,\pm)}|G\setminus v\rangle\label{eq:PZ},\\
P_v^{(Y,\pm)}&=&\tfrac{1}{2}|Y,\pm\rangle_v\otimes U_v^{(Y,\pm)}|c_v(G)\setminus v\rangle\label{eq:PY},\\
P_v^{(X,\pm)}&=&\tfrac{1}{2}|X,\pm\rangle_v\otimes U_v^{(X,\pm)}|\text{ep}_{(w,v)}(G)\setminus v\rangle\nonumber.\\\label{eq:PX}
\end{eqnarray}
The third equation above is defined with respect to an edge pivot $\text{ep}_{(w,v)}$ for any choice of $w\in N_G(v)$.
Note that each of the graph states $|G\setminus v\rangle$, $|c_v(G)\setminus v\rangle$, and $|\text{ep}_{(w,v)}(G)\setminus v\rangle$ obtained from Equations \ref{eq:PZ}, \ref{eq:PY}, and \ref{eq:PX} correspond to vertex minors of $G$.
The operations $U_v^{(Z,\pm)}$, $U_v^{(Y,\pm)}$, and $U_v^{(Z,\pm)}$ are defined in~\cite{hein2004multiparty} as follows:
\begin{eqnarray}
U_v^{(Z,+)}&=&\text{id}_v;\\
U_v^{(Z,-)}&=&\prod_{u\in N(v)}Z_u;\\
U_v^{(Y,+)}&=&\prod_{u\in N(v)}\sqrt{-iZ_u};\\
U_v^{(Y,-)}&=&\prod_{u\in N(v)}\sqrt{iZ_u};\\
U_v^{(X,+)}&=&\sqrt{iY_w}\prod_{u\in N(v)\setminus N(w)\setminus w}Z_u;\\
U_v^{(X,-)}&=&\sqrt{-iY_w}\prod_{u\in N(w)\setminus N(v)\setminus v}Z_u.
\end{eqnarray}

\begin{figure}[t]
\centering
\includegraphics[width=\linewidth,page=19]{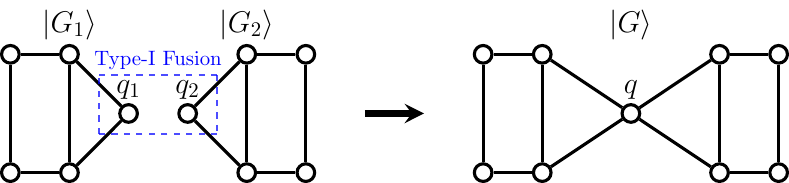}
\caption{
An example of a Type-I fusion between qubits from the same two graph states of Figure~\ref{fig:type-II_fusion}.
}
\label{fig:type-I_fusion}
\end{figure}

Although we use only Type-II fusions in this paper, Type-I fusions between graph states are also possible~\cite{browne2005resource}. Under this operation, qubits $q_1$ and $q_2$ from graph states $|G_1\rangle$ and $|G_2\rangle$ are combined into a single qubit $q$, inheriting all the neighbors of both $q_1$ and $q_2$.
What results is a single graph state $|G\rangle$, as shown in Figure~\ref{fig:type-II_fusion}.

\section{Formulas for Complete-Multipartie Graphs and Clique-Stars}
\label{app:formulas_for_CM_and_CS}

In~\cite{connolly2026local}, we derive a number of explicit formulas relating to complete-multipartite graphs and clique-stars.
The values in Table~\ref{tab:optimal_vs_split-fuse_comparison} are based on these formulas, but we excluded them from the main text due to their length.
These formulas include the size of the LC orbit, the number of edges in a minimal-edge representative of this orbit, and the maximum vertex degree $\Delta(G)$ in a minimal $\Delta(G)$ representative of this orbit. The latter two types of formulas fall into three cases each depending on the values of $k,n_1,\cdots,n_k$, the details of which are explained in~\cite{connolly2026local}, but for which use $n_j=\min\{n_1,\cdots,n_k\}$, $n_t=\min\{n_1,\cdots,n_k\}\setminus\{n_j\}$, and $n_\ell=\max\{n_1,\cdots,n_k\}$.

\begin{widetext}
\begin{eqnarray}
|{\mathcal O}(K_{n_1,\cdots,n_k})|&=&\underbrace{\prod_{i=1}^k(n_i+1)}_{\text{even products}}+\sum_{j=1}^k\left(\prod_{i=1\neq j}^k(n_i+1)\right)\label{eq:CM_orbit_size}\\
|{\mathcal O(CS^r_{n_1,\cdots,n_k})|}&=&\underbrace{\prod_{i=1}^k(n_i+1)}_{\text{odd products}}+\sum_{j=1}^k\left(\prod_{i=1\neq j}^k(n_i+1)\right)\label{eq:CS_orbit_size}\\
\min_{G\cong_{\text{LC}}K_{n_1,\cdots,n_k}}|E(G)|&=&\begin{cases}\frac{k(k-1)}{2}+\sum_{i=1}^k(n_i-1)&\text{(Case 1)}\\n_j(k-1)+\frac{n_j(n_j-1)}{2}+\sum_{i=1\neq j}^k(n_i-1)&\text{(Case 2)}\\n_j(k-1)+\sum_{i=1\neq j}^k(n_i-1)&\text{(Case 3)}\end{cases}\label{eq:CM_min_edges}\\
\min_{G\cong_{\text{LC}}CS^r_{n_1,\cdots,n_k}}|E(G)|&=&\begin{cases}n_j(k-1)+\sum_{i=1\neq j}^k(n_i-1)&\text{(Case 1)}\\\frac{k(k-1)}{2}+\sum_{i=1}^k(n_i-1)&\text{(Case 2)}\\n_j(k-1)+\frac{n_j(n_j-1)}{2}+\sum_{i=1\neq j}^k(n_i-1)&\text{(Case 3)}\end{cases}\label{eq:CS_min_edges}\\
\min_{G\cong_{\text{LC}}K_{n_1,\cdots,n_k}}\Delta G&=&\begin{cases}\min\begin{cases}n_\ell+k-2&\text{(Case 1)}\\\max\{(k-1+n_t),(n_\ell-1+n_j)\}&\text{(Case 2)}\\\max\{(n_j+k-2),(n_\ell-1+n_j)\}&\text{(Case 3)}\end{cases}&k\text{ is even}\\\min\begin{cases}n_\ell+n_j+k-3&\text{(Case 1)}\\\max\{(k-1),(n_\ell-1+n_j)\}&\text{(Case 2)}\\\max\{(n_j+n_t+k-3),(n_\ell-1+n_j)\}&\text{(Case 3)}\end{cases}&k\text{ is odd}\end{cases}\label{eq:CM_min_deltaG}\\
\min_{G\cong_{\text{LC}}CS^r_{n_1,\cdots,n_k}}\Delta G&=&\begin{cases}\min\begin{cases}n_\ell+n_j+k-3&\text{(Case 1)}\\\max\{(k-1),(n_\ell-1+n_j)\}&\text{(Case 2)}\\\max\{(n_j+n_t+k-3),(n_\ell-1+n_j)\}&\text{(Case 3)}\end{cases}&k\text{ is even}\\\min\begin{cases}n_\ell+k-2&\text{(Case 1)}\\\max\{(k-1+n_t),(n_\ell-1+n_j)\}&\text{(Case 2)}\\\max\{(n_j+k-2),(n_\ell-1+n_j)\}&\text{(Case 3)}\end{cases}&k\text{ is odd}\end{cases}\label{eq:CS_min_deltaG}
\end{eqnarray}
\end{widetext}

Equations~\ref{eq:CM_orbit_size} and \ref{eq:CS_orbit_size} are stated with respect to \textit{even} and \textit{odd} products in $\prod_{i=1}^k(n_i+1)$. By this, we mean the number of terms in each of the summands of the fully expanded form, such as $n_1n_2$ or $n_1n_4n_5$.

If we assume that $k,n_1,\cdots,n_k=O(n)$ are all on the order of $n$, observe that the leading terms in the formulas for the minimum number of edges are quadratic. This is in contrast to the number of CZ gates required by the split-fuse method as stated in Equation~\ref{eq:split_fuse_CZ}, which scales linearly with respect to $n$.
Hence, this should outperform a choice of optimal graph state in general as the number of qubits increases.

\section{Proof of Proposition~\ref{thm:reconstruction}}
\label{app:proof}

\begin{proof}
Let $G$ be any graph with a split decomposition described by $\textit{QASST}(G)=(Q_1,\cdots,Q_k)$.
Suppose we prepare intermediate graph states $|Q_1\rangle,\cdots,|Q_k\rangle$ corresponding to each of the quotient graphs in the QASST.
Our goal is to show how the target graph state $|G\rangle$ can be assembled from these intermediate graph states using Type-II fusions.

The edges of the underlying tree structure of $\textit{QASST}(G)$ are in bijection with the strong splits of $G$. Each pair of split-nodes joined by an edge in the QASST represents a complete bipartite subgraph of $G$, with the neighbors of each split-node corresponding to one half of the bipartition.
If $Q_i$ and $Q_j$ are quotient graphs adjacent in the QASST by a pair of split-nodes $s_i^j$ and $s_j^i$, a Type-II fusion between the qubits in $|Q_i\rangle$ and $|Q_j\rangle$ corresponding to $s_i^j$ and $s_j^i$ fully connects the neighbors of $s_i^j$ in $|Q_i\rangle$ with the neighbors of $s_j^i$ in $|Q_j\rangle$, removing $s_i^j$ and $s_j^i$ in the process but affecting no other connections.
Effectively, this operation rebuilds a complete bipartite subgraph, undoing the collapse introduced by the split decomposition.
The result is a new intermediate graph state $|Q_i'\rangle$ consisting of the remaining qubits from $|Q_i\rangle$ and $|Q_j\rangle$.

This process can be repeated for each edge in the QASST representing a connection between neighboring quotient graphs; each time, two intermediate graph states are fused into a larger intermediate graph state, consuming the pair of split-node qubits matching an edge in the QASST.
The procedure terminates when all pairs of split-nodes are consumed, exhausting all of the edges in $\textit{QASST}(G)$.
This occurs when all intermediate graph states have been fused into a single graph state corresponding to $|G\rangle$.
This process exactly parallels the recursive computation of the split decomposition of $G$ in reverse; replacing a pair of adjacent split-nodes in $\textit{QASST}(G)$ by all-to-all connections between neighbors is the inverse of collapsing a strong split in $G$ into a pair of connected split-nodes.
Furthermore, the order of the reconstruction is unimportant; the final resulting graph is uniquely determined by the split decomposition.
\end{proof}

\end{document}